\RequirePackage{amsmath}
\documentclass{llncs}

\usepackage{multicol}
\usepackage{amsfonts}
\usepackage{amssymb}
\usepackage{graphicx}
\usepackage{epic}
\usepackage{eepic}
\usepackage{epsfig,float}
\usepackage{verbatim}
\usepackage{pdfsync}
\usepackage{gastex}
\usepackage{geometry}
\newcommand{\Bsf}{\mathbf{B}}  
\newcommand{\Vsf}{\mathbf{T}^{\le 5}} 
\newcommand{\Wsf}{\mathbf{T}^{\ge 6}}

\newcommand{\Gsf}{\mathbf{G}^{\le 5}} 
\newcommand{\Msf}{\mathbf{H}^{\le 5}} 
\newcommand{\Hsf}{\mathbf{G}^{\ge 6}}

\renewcommand{\le}{\leqslant}
\renewcommand{\ge}{\geqslant}

\newcommand{\ol}{\overline}
\newcommand{\eps}{\varepsilon}
\newcommand{\emp}{\emptyset}

\newcommand{\Sig}{\Sigma}

\newcommand{\noin}{\noindent}

\newcommand{\bi}{\begin{itemize}}
\newcommand{\ei}{\end{itemize}}
\newcommand{\be}{\begin{enumerate}}
\newcommand{\ee}{\end{enumerate}}
\newcommand{\bd}{\begin{description}}
\newcommand{\ed}{\end{description}}
\newcommand{\bq}{\begin{quote}}
\newcommand{\eq}{\end{quote}}
\newcommand{\txt}[1]{\mbox{ #1 }}

\newcommand{\cA}{{\mathcal A}}

\newcommand{\cD}{{\mathcal D}}

\newcommand{\cN}{{\mathcal N}}

\newcommand{\cP}{{\mathcal P}}
\newcommand{\cR}{{\mathcal R}}

\newcommand{\cT}{{\mathcal T}}

\newcommand{\lraL}{{\mathbin{\approx_L}}}

\newcommand{\timg}{\mbox{rng}}

\newcommand{\qedb}{\hfill$\blacksquare$}

\title{Complexity of Suffix-Free Regular Languages}

\author{Janusz~Brzozowski\inst{1} \and Marek Szyku{\l}a \inst{2}}

\titlerunning{Complexity of Suffix-Free Languages}

\authorrunning{J. Brzozowski and M. Szyku{\l}a}   

\institute{David R. Cheriton School of Computer Science, University of Waterloo, \\
Waterloo, ON, Canada N2L 3G1\\
\{{\tt brzozo@uwaterloo.ca}\}
\and
Institute of Computer Science, University of Wroc{\l}aw,\\
Joliot-Curie 15, PL-50-383 Wroc{\l}aw, Poland\\
\{{\tt msz@cs.uni.wroc.pl}\}
}

\begin{document}
\maketitle
%\today
\begin{abstract}
We study various complexity properties of suffix-free regular languages.
The \emph{quotient complexity} of a regular language $L$ is the number of left quotients of $L$; this is the same as the \emph{state complexity} of $L$, which is the number of states in a minimal deterministic finite automaton (DFA) accepting $L$.
A regular language $L'$ is a \emph{dialect} of a regular language $L$ if it differs only slightly from $L$ (for example, the roles of the letters of $L'$ are a permutation of the roles of the letters of $L$).
The \emph{quotient complexity of an operation} on regular languages is the maximal quotient complexity of the result of the operation expressed as a function of the quotient complexities of the operands. 
A sequence $(L_k,L_{k+1},\dots)$ of regular languages in some class ${\mathcal C}$, where $n$ is the quotient complexity of $L_n$, is called a \emph{stream}.
A stream is \emph{most complex} in class ${\mathcal C}$ if its languages $L_n$ meet the complexity upper bounds for all basic measures, namely, they meet the quotient complexity upper bounds for star and reversal;  they have largest syntactic semigroups;  they  have the maximal numbers of atoms, each of which has maximal quotient complexity; and (possibly together with their dialects $L'_m$)  they  meet the quotient complexity upper bounds for boolean operations and product (concatenation).
It is known that there exist such most complex streams in the class of regular languages, in the class of prefix-free languages, and also in the classes of right, left, and two-sided ideals.
In contrast to this, we prove that there does not exist a most complex stream in the class of suffix-free regular languages.
However, we do exhibit one ternary suffix-free stream that meets the bound for product and whose restrictions to binary alphabets meet the bounds for star and boolean operations. 
We also exhibit a quinary stream that meets the bounds for boolean operations, reversal, size of syntactic semigroup, and atom complexities. 
Moreover, we solve an open problem about the bound for the product of two languages of quotient complexities $m$ and $n$ in the binary case by showing that it can be met for infinitely many $m$ and $n$.
Two transition semigroups play an important role for suffix-free languages:
semigroup $\Vsf(n)$ is a suffix-free semigroup that has maximal cardinality for $2\le n\le 5$, while semigroup
$\Wsf(n)$ has maximal cardinality for $n=2,3$, and $n\ge 6$.
We prove that all witnesses meeting the bounds for the star and the second witness in a product must have transition semigroups in $\Vsf(n)$. 
On the other hand, witnesses meeting the bounds for reversal, size of syntactic semigroup, and the complexity of atoms must have semigroups in $\Wsf(n)$. 
\medskip

\noin
{\bf Keywords:}
most complex, regular language, state complexity, suffix-free, syntactic complexity, transition semigroup
\end{abstract}

\section{Introduction}

We study complexity properties of suffix-free regular languages.
A much shorter preliminary version of these results appeared in~\cite{BrSz15b} without any proofs.
\smallskip

\noin{\bf Motivation}
The \emph{state complexity} $\kappa(L)$ of a regular language $L$ over an alphabet $\Sig$ is the number of states in a minimal deterministic finite automaton (DFA) with input alphabet $\Sig$ recognizing $L$. 
The state complexity of a regularity preserving unary operation $\circ$ on regular languages is the maximal value of $\kappa(L_n^\circ)$ as a function of $n$, where $L_n$ varies over all regular languages $L_n$ with state complexity $n$.
Similarly, the state complexity of a regularity preserving binary operation $\circ$ on regular languages is the maximal value of $\kappa(K_m\circ L_n)$ as a function of $m$ and $n$, where $K_m$ and $L_n$ vary over all regular languages of state complexities $m$ and $n$, respectively.
Of special interest are the state complexities of common operations on regular languages.
The state complexities of union, product (concatenation), (Kleene) star and reversal were studied by Maslov~\cite{Mas70} in 1970, but this work was not well known for many years. 
The paper by Yu, Zhuang, and Salomaa~\cite{YZS94} inspired considerable interest in these problems and much research has been done on this topic in the past 22 years.
The state complexity of an operation gives a worst-case lower bound on the time and space complexities of the operation. For this reason it has been studied extensively; see~\cite{Brz10a,HoKu11,Yu01} for additional references.

Consider, for example, the product of two regular languages. To find its state complexity we need to establish an upper bound for it, and find two languages of state complexities $m$ and $n$, respectively that meet this bound. It is known that $(m-2)2^n+2^{n-1}$ is a tight upper bound on product. The languages that meet this bound are called \emph{witnesses}.
In general, different witnesses have been used for the two arguments and for different operations. However, Brzozowski~\cite{Brz13} has shown that one  witness and its slightly modified version called a \emph{dialect} is sufficient to meet the bounds on all binary boolean operations, product, star, and reversal. The DFA of this witness is shown in Figure~\ref{fig:regular}.
Let the language recognized by this DFA be $L_n(a,b,c)$, and let $L_n(b,a,c)$ be the language of the DFA obtained from that of Figure~\ref{fig:regular} by interchanging the roles of $a$ and $b$. 
Then $L_m(a,b,c)$ and $L_n(b,a,c)$ meet the bound $mn$ for all binary boolean operations, $L_m(a,b,c)$ and $L_n(a,b,c)$ meet the bound $(m-1)2^n+2^{n-1}$ for product, and $L_n(a,b,c)$ meets the bounds $2^{n-1} +2^{n-2}$ for star and $2^n$ for reversal.

\begin{figure}[ht]
\unitlength 8.5pt\footnotesize
\begin{center}\begin{picture}(37,7)(0,4)
\gasset{Nh=1.8,Nw=3.7,Nmr=1.25,ELdist=0.2,loopdiam=1.5}
\node(0)(1,7){0}\imark(0)
\node(1)(8,7){1}
\node(2)(15,7){2}
\node[Nframe=n](3dots)(22,7){$\dots$}
\node(n-2)(29,7){$n-2$}
\node(n-1)(36,7){$n-1$}\rmark(n-1)
\drawloop(0){$c$}
\drawedge[curvedepth=.8,ELdist=.4](0,1){$a,b$}
\drawedge[curvedepth=.8,ELside=r](1,0){$b$}
\drawedge(1,2){$a$}
\drawloop(1){$c$}
\drawloop(2){$b,c$}
\drawedge(2,3dots){$a$}
\drawedge(3dots,n-2){$a$}
\drawloop(n-2){$b$}
\drawedge(n-2,n-1){$a$}
\drawedge[curvedepth=4.0,ELside=r,sxo=1,exo=-1](n-1,0){$a,c$}
\drawloop(n-1){$b,c$}
\end{picture}\end{center}
\caption{Minimal DFA of a most complex regular language.}
\label{fig:regular}
\end{figure}
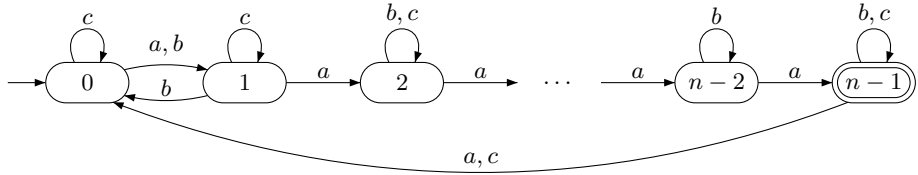

Although state complexity is a useful measure, it has some deficiencies. The language $L_n(a,b,c)$ and the language $\{a,b,c\}^{n-1}$ both have state complexity $n$, but the second language is intuitively much simpler. Secondly, all the bounds for common operations -- except reversal, which misses the bound by one state -- are also met by star-free languages~\cite{BrLi12}, where the class of star-free languages is the smallest class containing the finite languages and closed under boolean operations and product, but not star.
For these reasons Brzozowski~\cite{Brz13} suggested additional complexity measures: the maximal size of the syntactic semigroup of the language and the state complexities of atoms (discussed later). It turns out that the language defined in Figure~\ref{fig:regular} also meets these bounds. For these reasons this language has been called a \emph{most complex regular language}.
The added measure of the size of the syntactic semigroup distinguishes well between $L_{n}(a,b,c)$ and $\{a,b,c\}^{n-1}$ since the semigroup of the first language has $n^n$ elements, while the second has $n-1$ elements.

In previous work on state complexity it was always assumed that the two arguments in product and binary boolean operations are restricted to be over the same alphabet. But Brzozowski pointed out in 2016~\cite{Brz16} that operations on languages over different alphabets, \emph{unrestricted operations}, have higher tight upper bounds. It was shown that witnesses very similar to those of Figure~\ref{fig:regular} also meet these bounds~\cite{Brz16,BrSiJALC}, and thus are most complex for both restricted and unrestricted operations. For suffix-free languages the restricted and unrestricted complexities are the same.

Most complex languages are useful for testing the efficiency of systems.
Hence to check the maximal size of the objects that a system can handle, we can use most complex languages. It is certainly simpler to have just one or two worst-case examples.

The question now arises whether most complex languages also exist in subclasses of regular languages. A natural source of subclasses is obtained from the notion of convexity.
Convex languages were introduced by Thierrin~\cite{Thi73} and studied later in~\cite{AnBr09}. They can be defined with respect to arbitrary binary relations on an alphabet $\Sig^*$, but the relations ``is a prefix of'', ``is a suffix of'' and ``is a factor of'' turned out to be of considerable interest, where if $w=xyz\in \Sig^*$, then $x$ is a \emph{prefix} of $w$, $y$ is a \emph{factor}, and $z$ is a \emph{suffix}. 
A language $L$ is \emph{prefix-convex} if, whenever $w=xyz$ and $w$ and $x$ are in $L$, then so is $xy$. \emph{Factor-convex} and \emph{suffix-convex} languages are defined in a similar way.
Brzozowski began studying the complexity properties of these languages in~\cite{Brz10}.

The class of prefix-convex languages has four natural subclasses: \emph{right ideals} (languages $L$ satisfying  $L=L\Sig^*$), \emph{prefix-closed} languages (where $w$ in $L$ implies that every prefix of $w$ is in $L$), 
\emph{prefix-free} languages (where $w$ in $L$ implies that no prefix of $w$ other than $w$ is in $L$), and \emph{prefix-proper} languages that are prefix-convex but do not belong to any one of the three special subclasses.
Similarly, there are four subclasses of suffix-free languages: \emph{left ideals} (languages $L$ satisfying  $L=\Sig^*L$), \emph{suffix-closed} languages, 
\emph{suffix-free} languages, and \emph{suffix-proper} languages.
Finally, there are four subclasses of factor-free languages: \emph{two-sided ideals} (languages $L$ satisfying  $L=\Sig^*L\Sig^*$), \emph{factor-closed} languages, 
\emph{factor-free} languages, and \emph{factor-proper} languages.
Ideals appear in pattern matching~\cite{CrHa90}: if we are looking for all the words in a given text that begin with words in a pattern language $L$, then we are dealing with the right ideal $L\Sig^*$, and similar statements apply to left and two-sided ideals. Prefix-closed (respectively, suffix-closed, factor closed) languages are complements of right ideals (respectively, left ideals, two-sided ideals). Prefix-free (respectively, suffix-free, factor-free) languages, other than the language consisting of the empty word, are codes~\cite{BPR09}, and have many applications, particularly in cryptography, data compression and error correction.

Most complex left, right and two-sided ideals for restricted operations were found in~\cite{BDL17}, and unrestricted operations were added in~\cite{BrSiJALC}.
Most complex prefix-closed, and prefix-free languages were exhibited in~\cite{BrSi16AC} and proper prefix-convex languages in~\cite{BrSi16D}.
Most complex suffix-closed languages were found in~\cite{BrSi16L}.
In contrast to these results, the first example of a subclass of the class of regular languages that does not have a most complex language is that of suffix-free languages.
This result was reported in~\cite{BrSz15}, and is the subject of the present paper.
It is also known (unpublished result) that most complex proper suffix-convex languages do not exist.
\smallskip

\noin
{\bf Quotient Complexity} A basic complexity measure of a regular language $L$ over an alphabet $\Sig$ is the number $n$ of distinct left quotients of $L$, where a \emph{(left) quotient} of $L$ by a word $w\in\Sig^*$ is $w^{-1}L=\{x\mid wx\in L\}$. 
We denote the set of quotients of $L$ by $K=\{K_0,\dots,K_{n-1}\}$, where $K_0=L=\eps^{-1}L$ by convention.
Each quotient $K_i$ can be represented also as $w_i^{-1}L$, where $w_i\in\Sig^*$ is such that
$w_i^{-1}L=K_i$.
The number of quotients of $L$ is its \emph{quotient complexity}~\cite{Brz10a}.
Since the set of quotients of $L$ is the same as the number of states in a minimal DFA recognizing $L$, the quotient complexity of $L$ is the same as its state complexity; however, quotient complexity suggests language-theoretic methods, whereas state complexity deals with automata.
The quotient complexities of unary operation and binary operations are defined analogously to state complexities.

To establish the state/quotient complexity of a regularity preserving unary operation $\circ$ we need a sequence $(L_n, n\ge k)=(L_k,L_{k+1},\dots)$, called a \emph{stream}, of witness languages that meet this bound; here $k$ is usually some small integer because the bound may not apply for $n<k$.
The languages in a stream are normally defined in the same way, differing only in the parameter $n$. 
For example, the languages of the DFAs of Figure~\ref{fig:regular} define the stream
$(L_n(a,b,c)\mid n\ge 3)$.
Similarly, to establish the state/quotient complexity of a regularity preserving binary operation $\circ$ on regular languages we have to find two streams $(K_m, m \ge h)$ and $(L_n, n\ge k)$ of languages meeting this bound.

The state/quotient complexity of suffix-free languages was examined in~\cite{CmJi12,HaSa09,JiOl09}.

We also extend the notions of \emph{maximal complexity}, \emph{stream}, and \emph{witness} to DFAs. 

\noin
{\bf Syntactic Complexity}
A second measure of complexity of a regular language is its syntactic complexity.
Let $\Sig^+$ be the set of non-empty words of $\Sig^*$.
The \emph{syntactic semigroup} of $L$ is the set of equivalence classes of the Myhill congruence $\lraL$ on $\Sig^+$  defined by
$x~\lraL~y$ if and only if $uxv\in L  \Leftrightarrow uyv\in L\mbox { for all } u,v\in\Sig^*.$
The syntactic semigroup of $L$ is isomorphic to the \emph{transition semigroup} of a minimal DFA $\cD$ recognizing $L$~\cite{Pin97}, which is the semigroup of transformations of the state set of $\cD$ induced by non-empty words.
The \emph{syntactic complexity} of $L$ is the cardinality of its syntactic/transition semigroup.

Holzer and K\"onig~\cite{HoKo04}, and independently Krawetz, Lawrence and Shallit~\cite{KLS05} studied the syntactic complexity in the classes of unary and binary regular languages.
This problem was also solved for the classes of right ideals~\cite{BrSiJALC,BrYe11}, 
left ideals~\cite{BrSiJALC,BrSz14a,BrYe11}, two-sided ideals~\cite{BrSiJALC,BrSz14a,BrYe11}, 
prefix-free languages~\cite{BLY12}, and suffix-free languages~\cite{BLY12,BrSz15b}.

\noin
{\bf Complexities of Atoms}
A possible third measure of complexity of a regular language $L$ is the number and quotient complexities, which we call simply \emph{complexities}, of certain languages, called atoms, uniquely defined by $L$.
Atoms arise from an equivalence on $\Sig^*$ which is a left congruence refined by the Myhill congruence, where two words $x$ and $y$ are equivalent if 
 $ux\in L$ if and only if  $uy\in L$ for all $u\in \Sig^*$~\cite{Iva15}. 
 Thus $x$ and $y$ are equivalent if
 $x\in u^{-1}L  \Leftrightarrow y\in u^{-1}L$.
 An equivalence class of this relation is called an \emph{atom}~\cite{BrTa14} of $L$.
It follows that an atom is a non-empty intersection of complemented and uncomplemented quotients of $L$.
 The quotients of a language are unions of its atoms. 

\noin{\bf Terminology and Notation}
A \emph{deterministic finite automaton (DFA)} is defined as a quintuple
$\cD=(Q, \Sigma, \delta, q_0,F)$, where
$Q$ is a finite non-empty set of \emph{states},
$\Sig$ is a finite non-empty \emph{alphabet},
$\delta\colon Q\times \Sig\to Q$ is the \emph{transition function},
$q_0\in Q$ is the \emph{initial} state, and
$F\subseteq Q$ is the set of \emph{final} states.
We extend $\delta$ to a function $\delta\colon Q\times \Sig^*\to Q$ as usual.
A~DFA $\cD$ \emph{accepts} a word $w \in \Sigma^*$ if ${\delta}(q_0,w)\in F$. The language accepted by $\cD$ is denoted by $L(\cD)$. If $q$ is a state of $\cD$, then the language $L^q$ of $q$ is the language accepted by the DFA $(Q,\Sigma,\delta,q,F)$. 
A state is \emph{empty} if its language is empty. Two states $p$ and $q$ of $\cD$ are \emph{equivalent} if $L^p = L^q$. 
A state $q$ is \emph{reachable} if there exists $w\in\Sig^*$ such that $\delta(q_0,w)=q$.
A DFA is \emph{minimal} if all of its states are reachable and no two states are equivalent.
Usually DFAs are used to establish upper bounds on the quotient complexity of operations and also as witnesses that meet these bounds.

A \emph{nondeterministic finite automaton (NFA)} is a quintuple
$\cD=(Q, \Sigma, \delta, I,F)$, where
$Q$,
$\Sig$ and $F$ are defined as in a DFA, 
$\delta\colon Q\times \Sig\to 2^Q$ is the \emph{transition function}, and
$I\subseteq Q$ is the \emph{set of initial states}. 
An \emph{$\eps$-NFA} is an NFA in which transitions under the empty word $\eps$ are also permitted.

The \emph{quotient DFA} of a regular language $L$ with $n$ quotients is defined by
$\cD=(K, \Sigma, \delta_\cD, K_0,F_\cD)$, where 
$\delta_\cD(K_i,w)=K_j$ if and only if $w^{-1}K_i=K_j$, 
and $F_\cD=\{K_i\mid \eps \in K_i\}$.
To simplify the notation, without loss of generality we use the set $Q_n=\{0,\dots,n-1\}$ of subscripts of quotients as the set of states of $\cD$; then $\cD$ is denoted by
$\cD=(Q_n, \Sigma, \delta, 0,F)$, where $\delta(p,w)=q$  if $\delta_\cD(K_p,w)=K_q$, and $F$ is the set of subscripts of quotients in $F_\cD$. 
The quotient DFA of $L$ is unique and it is isomorphic to each complete minimal DFA of $L$.

A \emph{transformation} of $Q_n$ is a mapping $t\colon Q_n\to Q_n$.
The \emph{image} of $q\in Q_n$ under $t$ is denoted by $qt$.
The \emph{range} of  $t$ is $\timg(t)=\{q\in Q_n\mid pt=q \text{ for some } p\in Q_n\}$.
In any DFA, each letter $a\in \Sig$ induces a transformation $\delta_a$ of the set $Q_n$ defined by $q\delta_a=\delta(q,a)$.
By a slight abuse of notation we use the letter $a$ to denote the transformation it induces; thus we write $qa$ instead of $q\delta_a$.
We also extend the notation to sets of states: if $P\subseteq Q_n$, then $Pa=\{pa\mid p\in P\}$.
If $s,t$ are transformations of $Q$, their composition is denoted by $s \ast t$ and defined by
$q(s \ast t)=(qs)t$; the $\ast$ is usually omitted.
Then also for a word $w = a_1 \cdots a_k$, $\delta_w$ denotes the transformation $\delta_{a_1} \cdots \delta_{a_k}$ induced by $w$.
Let $\cT_{Q_n}$ be the set of all $n^n$ transformations of $Q_n$; then $\cT_{Q_n}$ is a monoid under composition. 

For $k\ge 2$, a transformation (permutation) $t$ of a set $P=\{q_0,q_1,\ldots,q_{k-1}\} \subseteq Q$ is a \emph{$k$-cycle}
if $q_0t=q_1, q_1t=q_2,\ldots,q_{k-2}t=q_{k-1},q_{k-1}t=q_0$.
This $k$-cycle is denoted by $(q_0,q_1,\ldots,q_{k-1})$.
A~2-cycle $(q_0,q_1)$ is called a \emph{transposition}.
 A transformation that changes only one state $p$ to a state $q\neq p$ is denoted by $(p\to q)$.
A transformation mapping a subset $P$ of $Q$ to a single state and acting as the identity on $Q\setminus P$ is denoted by $(P\to q)$.
We also denote by $[q_0,\dots,q_{n-1}]$ the transformation that maps $p\in \{0,\dots,n-1\}$ to $q_p$.

We now define dialects of languages and DFAs following~\cite{BDL17}.
Let $\Sig=\{a_1,\dots,a_k\}$ be an alphabet; we assume that its elements are ordered as shown.
Let $\pi$ be a \emph{partial permutation} of $\Sig$, that is, a partial function $\pi \colon \Sig \rightarrow \Gamma$ where $\Gamma \subseteq \Sig$, for which there exists $\Delta \subseteq \Sig$ such that $\pi$ is bijective when restricted to $\Delta$ and  undefined on $\Sig \setminus \Delta$. We denote undefined values of $\pi$ by the symbol ``$-$''.

If $L$ is a language over $\Sig$, we denote it by $L(a_1,\dots,a_k)$ to stress its dependence on $\Sig$.
If $\pi$ is a partial permutation, let $s_\pi$ be the language substitution  defined as follows: 
 for $a\in \Sig$, 
$a \mapsto \{\pi(a)\}$ when $\pi(a)$ is defined, and $a \mapsto \emp$ when $\pi(a)$ is not defined.
For example, if $\Sig=\{a,b,c\}$, $L(a,b,c)=\{a,b,c\}^*\{ab, acc\}$, and $\pi(a)=c$, $\pi(b)=-$, and $\pi(c)=b$, then $s_\pi(L)= \{b,c\}^*\{cbb\}$.
In other words, the letter $c$ plays the role of $a$, and $b$ plays the role of $c$.
A \emph{permutational dialect} of $L(a_1,\dots,a_k)$ is a language of the form 
$s_\pi(L(a_1,\dots,a_k))$, where $\pi$ is a partial permutation of $\Sig$; this dialect is denoted by
$L(\pi(a_1),\dotsc,\pi(a_k))$.
If the order on $\Sig$ is understood, we use $L(\Sig)$ for $L(a_1,\dots,a_k)$ and $L(\pi(\Sig))$
for $L(\pi(a_1),\dotsc,\pi(a_k))$.
Undefined values appearing at the end of the alphabet are omitted. For example, if $\Sigma=\{a,b,c,d\}$ then we write $L(a,b)$ instead of $L(a,b,-,-)$.

Let $\Sig=\{a_1,\dots,a_k\}$,  and 
let $\cD = (Q,\Sig,\delta,q_1,F)$ be a DFA; we denote it by
$\cD(a_1,\dots,a_k)$ to stress its dependence on $\Sig$.
If $\pi$ is a partial permutation, then the \emph{permutational dialect} 
$$\cD(\pi(a_1),\dotsc,\pi(a_k))$$ of
$\cD(a_1,\dots,a_k)$ is obtained by changing the alphabet of $\cD$ from $\Sig$ to $\pi(\Sig)$, and modifying $\delta$ so that in the modified DFA 
$\pi(a_i)$ induces the transformation induced by $a_i$  in the original DFA; thus $\pi(a_i)$ plays the role of $a_i$.
One verifies that if the language $L(a_1,\dots,a_k)$ is accepted by DFA $\cD(a_1,\dots,a_k)$, then
$L(\pi(a_1),\dotsc,\pi(a_k))$ is accepted by $\cD(\pi(a_1),\dotsc,\pi(a_k))$.

In the sequel we refer to permutational dialects simply as \emph{dialects}.

\smallskip

\noin{\bf Contributions}\\         
\hglue 15 pt 1. We prove that a  most complex stream of suffix-free languages does not exist. This is in contrast with the existence of streams of most complex regular languages, right, left, and two-sided ideals, prefix-free and proper prefix-convex languages. \\
\hglue 15 pt 2.
We exhibit a single ternary witness that meets the bounds for star, product, and boolean operations.\\
\hglue 15 pt 3. We exhibit a single quinary witness that meets the bounds for boolean operations, reversal, number of atoms, syntactic complexity, and quotient complexities of atoms.\\
\hglue 15 pt 4. We show that when $m,n\ge 6$ and $m-2$ and $n-2$ are relatively prime, there are binary witnesses that meet the bound $(m-1)2^{n-2}+1$ for  product. \\
\hglue 15 pt 5. We prove that any witness DFA for star and any second witness DFA for product must have transition semigroups  that are subsemigroups of the suffix-free semigroup of transformations $\Vsf(n)$ which has maximal cardinality for $2\le n\le 5$; that the witness DFAs for reversal, syntactic complexity and quotient complexities of atoms must have transition semigroups  that  are subsemigroups of the suffix-free semigroup of transformations 
$\Wsf(n)$ which has maximal cardinality for $n=2,3$ and $n\ge 6$; and that the witness DFAs for boolean operations can have transition semigroups  that  are subsemigroups of $\Vsf \cap \Wsf$.

\section{Suffix-Free Transformations}

In this section we discuss some properties of suffix-free languages with emphasis on their syntactic semigroups as represented by the transition semigroups of their quotient DFAs.
We assume that our basic set is always $Q_n=\{0,\dots,n-1\}$.

\subsection{Suffix-Free Languages}

Let $\cD_n=(Q_n, \Sigma, \delta, 0,F)$ be the quotient DFA of a suffix-free language $L$, and let $T_n$ be its transition semigroup.
For any transformation $t$ of $Q_n$, the sequence $(0,0t,0t^2,\dots)$ is called the \emph{$0$-path} of $t$.
Since $Q_n$ is finite, there exist $i,j$ such that $0,0t,\dots,0t^i,0t^{i+1},\dots,0t^{j-1}$ are distinct but $0t^j=0t^i$.
The integer $j-i$ is the \emph{period} of $t$ and if $j-i=1$, $t$ is \emph{initially aperiodic}.
The following properties of suffix-free languages are known~\cite{BLY12,HaSa09}:

\begin{lemma}
\label{lem:sf}
If $L$ is a regular suffix-free language, then
\be
\item
There exists $w \in \Sig^*$ such that $w^{-1}L=\emp$; hence $\cD_n$ has an empty state, which is state $n-1$ by convention.
\item
For $w,x\in \Sig^+$, if $w^{-1}L\neq \emp$, then $w^{-1}L\neq (xw)^{-1}L$.
\item
If $L\neq \emp$ and $w^{-1}L=L$, then $w=\eps$. 
\item
For any $t\in T_n$, the $0$-path of $t$ in $\cD_n$ is aperiodic and ends in $n-1$. 
\ee
\end{lemma}

Property~3 is  known as the \emph{non-returning} property~\cite{HaSa09} and also as \emph{unique reachability}~\cite{Brz10b}.

An (unordered) pair $\{p,q\}$ of distinct states in $Q_n \setminus \{0,n-1\}$ is \emph{colliding} (or $p$ \emph{collides} with $q$) in $T_n$ if there is a transformation $t \in T_n$ such that $0t = p$ and $rt = q$ for some $r \in Q_n \setminus \{0,n-1\}$. 
A pair of states is \emph{focused} by a transformation $u$ of $Q_n$ if $u$ maps both states of the pair to a single state $r \not\in \{0,n-1\}$. We then say that $\{p,q\}$ is \emph{focused to state $r$}.
If $L$ is a suffix-free language, then from Lemma~\ref{lem:sf}~(2) it follows that if $\{p,q\}$ is colliding in $T_n$, there is no transformation $t' \in T_n$ that focuses $\{p,q\}$.
So colliding states can be mapped to a single state by a transformation in $T_n$ only if that state is the empty state $n-1$.

Following~\cite{BLY12}, for $n \ge 2$, we let
$$\Bsf(n) = \{ t \in \cT_{Q} \mid 0 \not\in \timg(t), \; (n-1)t = n-1, \txt{and for all} j\ge 1,
 \hspace{2.5cm}$$ 
$$0t^j = n-1 \txt{or} 0t^j \neq qt^j,~\forall q \txt{such that} 0 < q < n-1\}.$$ 
\begin{example}
We have $\Bsf(2)=\{[1,1]\}$ and $\Bsf(3)= \{ [1,2,2], [2,1,2], [2,2,2]  \}$.
For $n=4$, there are 17 transformations satisfying $0t\neq qt$.
However,
if $t=[1,2,2,3]$ or $t=[2,1,1,3]$, then $0t^2=1t^2$, 
which violates the third condition for $\Bsf$;
hence $\Bsf(4)$
has 15 elements.
The cardinality of $\Bsf(5)$ is 115.
\qedb
\end{example}

\begin{proposition}[\cite{BLY12}]
\label{prop:sf}
If $L$ is a regular language with $\cD_n = (Q_n, \Sig, \delta, 0, F)$ as its minimal DFA and syntactic semigroup $T_L$, then the following hold:
\be
\item If $L$ is suffix-free, then $T_L$ is a subset of $\Bsf(n)$.
\item If $L$ has the empty quotient, only one final quotient, and $T_L \subseteq \Bsf(n)$, then $L$ is suffix-free.
\ee
\end{proposition}

Since the transition semigroup of a minimal DFA of a suffix-free language must be a subsemigroup of $\Bsf(n)$, 
the cardinality of $\Bsf(n)$ is an upper bound on the syntactic complexity of suffix-free regular languages with quotient complexity $n$.
This upper bound, however, cannot be reached since $\Bsf$ is not a semigroup for $n\ge 4$: We have $s=[1,2,n-1,\dots,n-1]$ and $t=[n-1,2,2,\dots,2,n-1]$ in $\Bsf(n)$, but
$st=[2,2,n-1,\dots,n-1]$ is not in $\Bsf(n)$.

\subsection{Semigroups $\Vsf(n)$ with Maximal Cardinality when $2\le n\le 5$}

For $n \ge 2$, let
$$\Vsf(n) = \{t \in \Bsf(n) \mid \txt{for all} p, q \in Q_n \txt{where} p \neq q, 
 \txt{either} pt = qt = n-1 \txt{or} pt \neq qt\}.$$

\begin{proposition}\label{prop:W5}
For $n\ge 4$, the semigroup $\Vsf(n)$ is generated by the following set $\Msf(n)$ of transformations of $Q$:
\bi
\item
${a}\colon (0\to n-1)(1,\dots,n-2)$,
\item
${b}\colon (0\to n-1)(1,2)$, 
\item
for $1\le p \le n-2$, $c_p\colon (p\to n-1)(0\to p)$.
\ei
For $n=4$, $a$ and $b$ coincide, and so $\Msf(4)=\{a, {c_1},{c_2}\}$. Also,
$\Msf(3)=\{a,{c_1}\}=\{[2,1,2], [1,2,2]\}$ and $\Msf(2)=\{{c_1}\}=\{[1,1]\}$. 

\end{proposition}
\begin{proof}
It was proved in~\cite{BLY12} that for $n \ge 4$, the semigroup $\Vsf(n)$ is generated by the following set ${\Gsf}(n)$ of  $n$  transformations of $Q_n$: $a$ and $b$ as above, and $c'_p$ 
for $1\le p\le n-2$, defined by
	$q{c'_p}=q+1$ for $q=0,\dots, p-1$, 
	$p{c'_p}=n-1$, 
and	$q{c'_p}=q$ for $q=p+1\dots, n-1$.
Since $\Msf(n)\subseteq \Vsf$, the semigroup generated by $\Msf(n)$ is a subsemigroup of $\Vsf(n)$.
So it is sufficient to show that every transformation in $\Gsf(n)$ can be generated by $\Msf(n)$.
Transformation $c'_p$ changes $\{0,1,\dots,p-1\}$ to $\{1,2,\dots,p\}$, $p$ is mapped to $n-1$, and $\{p+1,p+2,\dots,n-2)$ is mapped to itself. 
The image of $Q_n\setminus\{p,n-1\}$ is thus $\{1,2,\dots,p,p+1,\dots n-2\}$.
The image of $Q_n\setminus\{p,n-1\}$ under $c_p$ is $\{p,1,2,\dots,p-1,p+1,\dots,n-2\}$. However,
since transformations $a$ and $b$ restricted to $Q_n\setminus\{0,n-1\}$ generate all permutations of $Q_n\setminus\{0,n-1\}$,
$\{p,1,2,\dots,p-1,p+1,\dots,n-2\}$ can be transformed to $\{1,2,\dots,p,p+1,\dots n-2\}$.
Hence the claim holds.
\end{proof}

From now on we use the transformations of Proposition~\ref{prop:W5} for $\Vsf(n)$. A DFA using these transformations  is illustrated in Figure~\ref{fig:witness5}.

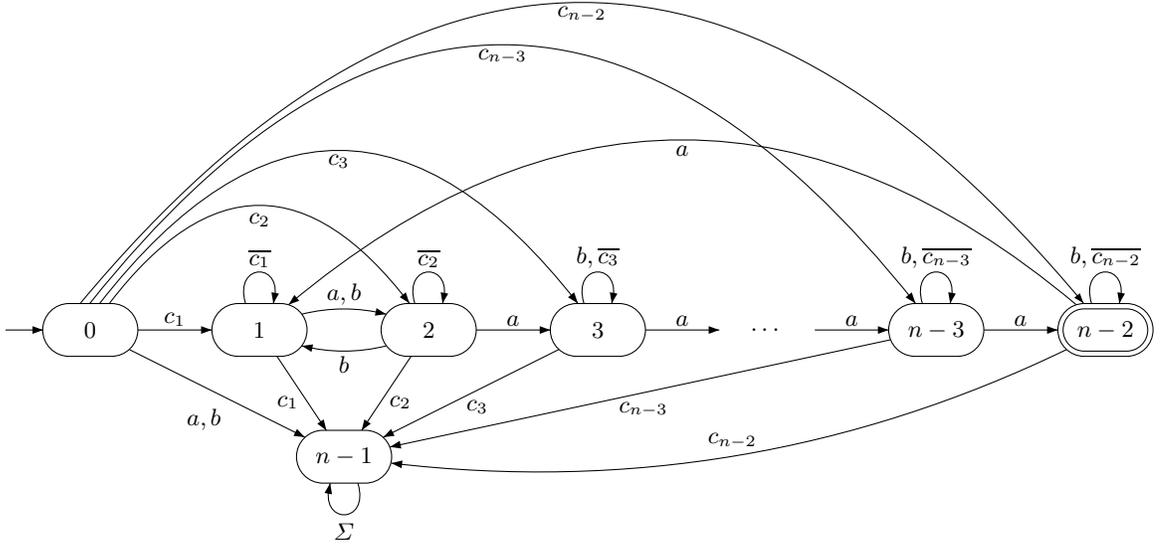
\begin{figure}[ht]
\unitlength 8pt
\begin{center}\begin{picture}(40,24)(4,-3)
\gasset{Nh=2.5,Nw=4.5,Nmr=1.25,ELdist=0.2,loopdiam=1.5}
{\small
\node(0)(1,7){0}\imark(0)
\node(1)(9,7){1}
\node(2)(17,7){2}
\node(3)(25,7){3}
\node[Nframe=n](3dots)(33,7){$\dots$}
\node(n-3)(41,7){$n-3$}
\node(n-2)(49,7){$n-2$}\rmark(n-2)
\node(n-1)(13,1){$n-1$}

\drawedge(0,1){$c_1$}
\drawedge[curvedepth=1](1,2){$a,b$}
\drawloop(1){$\ol{c_1}$}
\drawedge[curvedepth=1,ELdist=.25](2,1){$b$}
\drawloop(2){$\ol{c_2}$}
\drawedge(2,3){$a$}
\drawloop(3){$b,\ol{c_3}$}
\drawedge(3,3dots){$a$}
\drawedge(3dots,n-3){$a$}
\drawedge(n-3,n-2){$a$}
\drawloop(n-3){$b,\ol{c_{n-3}}$}
\drawloop(n-2){$b,\ol{c_{n-2}}$}
\drawedge[curvedepth=5.9,ELdist=-1](0,2){$c_2$}
\drawedge[curvedepth=-9,ELdist=.3](n-2,1){$a$}
\drawedge[curvedepth=8.5,sxo=-.5,ELside=r](0,3){$c_3$}
\drawedge[curvedepth=13.5,sxo=-1.,ELside=r](0,n-3){$c_{n-3}$}
\drawedge[curvedepth=15.5,sxo=-1.5,ELside=r](0,n-2){$c_{n-2}$}
\drawedge[ELdist= -2.2](0,n-1){$a,b$}
\drawedge[ELside=r](1,n-1){$c_1$}
\drawedge(2,n-1){$c_2$}
\drawedge(3,n-1){$c_3$}
\drawedge(n-3,n-1){$c_{n-3}$}
\drawedge[curvedepth=3,ELside=r](n-2,n-1){$c_{n-2}$}
\drawloop[loopangle=270,ELdist=0.4](n-1){$\Sigma$}
}
\end{picture}\end{center}
\caption{DFA with $\Vsf(n)$ as its transformation semigroup;
$\ol{c_p} = \{c_1,\dots,c_{n-2}\} \setminus \{c_p\}$.
}
\label{fig:witness5}
\end{figure}

\begin{proposition}
For $n\ge 2$, $\Vsf(n)$ is the unique maximal semigroup of a suffix-free language in which all possible pairs of states are colliding.
\end{proposition}
\begin{proof}
For each pair $p,q \in Q \setminus \{0,n-1\}$, $p\neq q$, there is a transformation $c_p \in \Vsf(n)$ with $0c_p = p$ and $qc_p = q$.
Thus all pairs are colliding.
If all pairs are colliding, then for each $p,q \in Q \setminus \{n-1\}$, there is no transformation $t$ with $pt = qt \neq n-1$,  for this would violate suffix-freeness.
By definition, $\Vsf(n)$ has all other transformations that are possible for a suffix-free language, and hence is unique.
\end{proof}

\begin{proposition}
For $n\ge 5$, the number $n$ of generators of $\Vsf(n)$ cannot be reduced.
\end{proposition}
\begin{proof}
If a generator $t$ maps $0$ to $p \in \{1,\ldots,n-2\}$, then it must also map a state $q \in \{1,\ldots,n-2\}$ to $n-1$, since the $0$-path of $t$ is aperiodic and ends in $n-1$.
Thus $t$ cannot generate a permutation of $\{1,\ldots,n-2\}$.
Since $\Vsf(n)$ has all transformations that map $0$ to $n-1$, permute  $\{1,\ldots,n-2\}$ and fix $n-1$, we need two generators, say $a$ and $b$, with $0 a=0b= n-1$ to induce all the permutations.
For $p \in \{1,\ldots,n-2\}$, consider the transformation ${c_p}$ of Proposition~\ref{prop:W5}.
We need a generator, say $a_p$, which maps $0$ to a state $p$ in $\{1,\ldots,n-2\}$, and a state $q$ from $\{1,\ldots,n-2\}$ to $n-1$.
But since $c_p$ maps only two states, $q$ and $n-1$, to $n-1$ and $\{1,\ldots,n-2\} \subset \timg({c_p})$, the other generators involved in the composition of generators that induces ${c_p}$ do not map any state from $\{1,\ldots,n-2\}$ to $n-1$.
Since there are $n-2$ distinct transformations $c_p \in \Vsf(n)$, one for each $p \in \{1,\ldots,n-2\}$, we need at least $n-2$ generators $a_p$.
This gives $n$ generators in total.
\end{proof}

\begin{example}
$\Vsf(4)$ has 13 elements.
All transitions of $\Bsf(4)$ are present except $[3,1,1,3]$ and $[3,2,2,3]$.
The semigroup $\Vsf(5)$ has 73 elements.
\qedb
\end{example}

Semigroups $\Vsf(n)$ are suffix-free semigroups that have maximal cardinality when $2\le n\le 5$~\cite{BLY12}.

\subsection{Semigroups $\Wsf(n)$ with Maximal Cardinality when $n=2,3$ and $n\ge 6$}

For $n \ge 2$, let 
$$\Wsf(n) = \{ t \in \Bsf(n) \mid 0t = n-1, \txt{or} qt = n-1 ~\forall~q \txt{such that} 1 \le q \le n-2 \}.$$ 
\begin{proposition}[\cite{BrSz15}]
\label{prop:Wsf} 
For $n \ge 4$, $\Wsf(n)$ is a semigroup contained in $\Bsf(n)$, its cardinality is 
$(n-1)^{n-2} + (n-2),$
and it is generated by the set $\Hsf(n)$ of the following transformations:
\bi
\item
$a\colon (0\to n-1)(1,\dots,n-2)$;
\item
$b\colon (0\to n-1) (1,2)$;
\item
$c\colon (0\to n-1) (n-2\to 1)$;
\item
$d\colon (\{0,1\}\to n-1)$;
\item
$e\colon (Q\setminus\{0\}\to n-1)(0\to 1)$.
\ei
For $n=4$, $a$ and $b$ coincide, and so $\Hsf(4)=\{a,c,d,e\}$.
Also $\Hsf(3)=\{a,e\}=\{[2,1,2],[1,2,2]\}$ and
$\Hsf(2)=\{e\}=\{[1,1]\}$.
\end{proposition}
 
 \begin{example}
$\Wsf(4)$ has 11 elements.
 All transitions of $\Bsf(4)$ are present except $[1,2,3,3]$, $[1,3,2,3]$, $[2,1,3,3]$ and $[2,3,1,3]$. Semigroup $\Wsf(5)$ has size 67.
 \qedb
\end{example}
A DFA using the transformations of Proposition~\ref{prop:Wsf} is shown in Figure~\ref{fig:witness}.

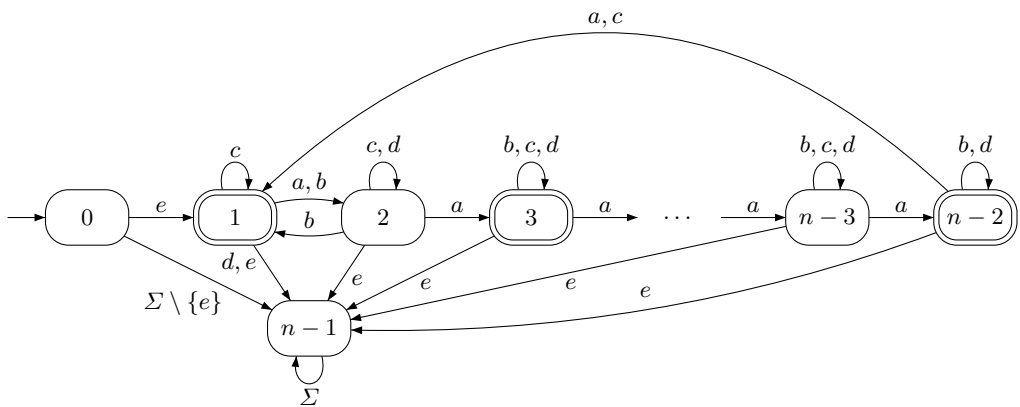
\begin{figure}[ht]
\unitlength 7pt\small
\begin{center}\begin{picture}(40,18)(5,-2)
\gasset{Nh=3.0,Nw=4.5,Nmr=1.25,ELdist=0.3,loopdiam=1.5}
\node(0)(1,7){0}\imark(0)
\node(1)(9,7){1}\rmark(1)
\node(2)(17,7){2}
\node(3)(25,7){3}\rmark(3)
\node[Nframe=n](3dots)(33,7){$\dots$}
\node(n-3)(41,7){$n-3$}
\node(n-2)(49,7){$n-2$}\rmark(n-2)
\node(n-1)(13,1){$n-1$}
\drawedge(0,1){$e$}
\drawedge[curvedepth=1](1,2){$a,b$}
\drawloop(1){$c$}
\drawedge[curvedepth=1,ELdist=-1.3](2,1){$b$}
\drawloop(2){$c,d$}
\drawedge(2,3){$a$}
\drawedge(3,3dots){$a$}
\drawedge(3dots,n-3){$a$}
\drawedge(n-3,n-2){$a$}
\drawedge[curvedepth=-10,ELside=r](n-2,1){$a,c$}
\drawloop(3){$b,c,d$}
\drawloop(n-3){$b,c,d$}
\drawloop(n-2){$b,d$}
\drawedge[ELside=r](0,n-1){$\Sigma\setminus\{e\}$}
\drawedge[ELside=r,ELdist=0.1,ELpos=30](1,n-1){$d,e$}
\drawedge(2,n-1){$e$}
\drawedge(3,n-1){$e$}
\drawedge(n-3,n-1){$e$}
\drawedge[curvedepth=2.0,ELdist=-1.3](n-2,n-1){$e$}
\drawloop[loopangle=270](n-1){$\Sigma$}
\end{picture}\end{center}
\caption{DFA with $\Wsf(n)$ as its transformation semigroup for the case when $n$ is odd.}
\label{fig:witness}
\end{figure}

Semigroups $\Wsf(n)$ are suffix-free semigroups that have maximal cardinality when $n\ge 6$~\cite{BrSz15}.

\section{Witnesses with Transition Semigroups in $\Vsf(n)$}

In this section we consider DFA witnesses whose transition semigroups are subsemigroups of $\Vsf(n)$. 
We show that there is one witness that satisfies the bounds for star, product and boolean operations.
\begin{definition}
\label{def:D5}
For $n\ge 6$, 
we define the DFA 
$\cD_n =(Q_n,\Sig,\delta,0,\{1\}),$
where $Q_n=\{0,\ldots,n-1\}$, $\Sig=\{a,b,c\}$, 
and $\delta$ is defined by the transformations
\bi
\item
$ a \colon (0 \to n-1) (1,2,3) (4,\dots,n-2) $,
\item
$ b \colon (2 \to n-1) (1 \to 2) (0 \to 1) (3,4) $, 
\item
$ c \colon (0 \to n-1) (1,\dots, n-2)  $.
\ei
\end{definition}

\begin{theorem}[Star, Product, Boolean Operations]
\label{thm:witness5}
Let $\cD_n(a,b,c)$ be the DFA of Definition~\ref{def:D5}, and let the language it accepts be
$L_n(a,b,c)$. 
For $n\ge 6$, $L_n$ and its permutational dialects meet the bounds for star, product and boolean operations as follows:
\be
\item
$L_n^*(a,b,-)$ meets the bound $2^{n-2}+1$.
{\rm [Cmorik and Jir\'askov\'a~\cite{CmJi12}]}
\item
$L_m(a,b,c)\cdot L_n(b,c,a)$ meets the bound $(m-1)2^{n-2}+1$.
\item
$L_m(a,b,-)$ and $L_n(-,b,a)$ meet the bounds $mn - (m+n-2)$  for union and symmetric difference, $mn - 2(m+n-3)$ for intersection, and $mn - (m+2n-4)$ for difference.
\ee
\end{theorem}

The claim about the star operation was proved in~\cite{CmJi12}. 
We add a result about the transition semigroup of the star witness and prove the remaining two claims in this section.

\subsection{Star}

In 2009 Han and Salomaa~\cite{HaSa09} showed that the language of a DFA over a four-letter alphabet meets the bound $2^{n-2}+1$ for the star operation for $n\ge 4$.
The transition semigroup of this DFA is a subsemigroup of $\Vsf(n)$.
In 2012 Cmorik and Jir\'askov\'a~\cite{CmJi12} showed that for $n\ge 6$ a binary alphabet $\{a,b\}$ suffices. 
The transition semigroup of this DFA is again a subsemigroup of $\Vsf(n)$.
We prove that these are special cases of the following general result:

\begin{theorem}\label{thm:star_witness_subsemigroup}
For $n\ge 4$, the transition semigroup of a minimal DFA $\cD(Q_n,\Sig,\delta,0,F)$ of a suffix-free language $L$ that meets the bound $2^{n-2}+1$ for the star operation is a subsemigroup of $\Vsf(n)$ and is not a subsemigroup of $\Wsf(n)$.
\end{theorem}
\begin{proof} 
To show that the transition semigroup of $\cD$ is a subsemigroup of $\Vsf(n)$, it suffices to show that every pair of states is colliding.

We construct an NFA $\cN$ for $L_n^*$ by making 0 a final state in $\cD$ -- this is possible since $0$ is uniquely reachable -- and adding an empty-word transition from every final state to 0. We then determinize $\cN$ using the subset construction to get a DFA $\cD^*$ for $L_n^*$. The states of $\cD^*$ are sets of states of $\cD$.

Consider a subset $S \subseteq Q_n$. We can assume that $n-1 \in S$, since $S$ and $S \cup \{n-1\}$ cannot be distinguished.
If $S \neq \{0,n-1\}$ and $S \neq \{n-1\}$, then $S$ can be reached only if $0 \in S$ and $S \cap F \neq \emptyset$, or $0 \not\in S$ and $S \cap F = \emptyset$, because by the construction for star in $\cN$ there is an $\varepsilon$-transition from every final state to the initial state $0$ and no transformation fixes $0$.
Thus, to meet the bound $2^{n-2}+1$, for each possible subset of $\{1,\ldots,n-2\}$ there must be a reachable subset $S$ containing that subset.

Suppose that $p,q \in \{1,\ldots,n-2\}$ are not colliding, that is, there is no transformation $t$ with $0t=p$ and $q \in \timg(t)$.
Consider $S$ that contains both $p$ and $q$.  Since the  $\eps$-transitions from final states to initial state 0 are the only sources of nondeterminism, $S$ must be reached from a subset $S'$ containing $0$ and $q' \in \{1,\ldots,n-2\}$ by a transformation $t$ with $0t = p$ and $q't = q$ (or $0t = q$ and $q't = p$), which contradicts that $p,q$ are not colliding.

Since transformations causing colliding pairs are not in $\Wsf(n)$, the transition semigroup of $\cD$ cannot be a subsemigroup of $\Wsf(n)$.
\end{proof}

\subsection{Product}
\label{subsec:product}

To avoid confusing the states of the two DFAs in a product, we label the states of the first DFA differently.
Let $\cD'_m=\cD'_m(a,b,c)=(Q'_m, \Sig,\delta',0',\{ 1' \})$, where 
$Q'_m=\{0',\ldots,(m-1)'\}$, and $\delta'(q',x)=p'$ if $\delta(q,x)=p$, and
let $\cD_n=\cD_n(b,c,a)$.
We use the standard  construction of the $\eps$-NFA $\cN$ for the product: the final state of $\cD'_m$ becomes non-final, and  an $\eps$-transition is added from that state to the initial state of $\cD_n$.
This is illustrated in Figure~\ref{fig:product3} for $m=9, n=8$.

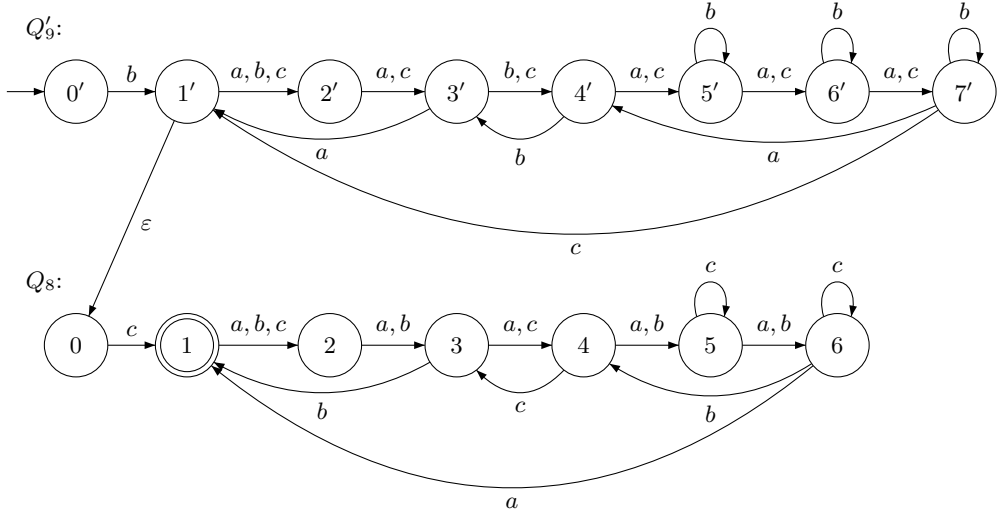
\begin{figure}[bht]
\unitlength 12pt
\begin{center}\begin{picture}(31,18)(0,-6)
\gasset{Nh=2.,Nw=2,Nmr=1.5,ELdist=0.3,loopdiam=1.0}
\node[Nframe=n](Qm)(1,10){$Q'_{9} $:}
\node(0')(2,8){$0'$}\imark(0')
\node(1')(5.5,8){$1'$}
\node(2')(10,8){$2'$}
\node(3')(14,8){$3'$}
\node(4')(18,8){$4'$}
\node(5')(22,8){$5'$}
\node(6')(26,8){$6'$}
\node(7')(30,8){$7'$}
\drawloop(5'){$b$}
\drawloop(6'){$b$}
\drawloop(7'){$b$}
\drawedge(0',1'){$b$}
\drawedge(1',2'){$a,b,c$}
\drawedge(2',3'){$a,c$}
\drawedge(3',4'){$b,c$}
\drawedge(4',5'){$a,c$}
\drawedge(5',6'){$a,c$}
\drawedge(6',7'){$a,c$}

\drawedge[curvedepth=1.5](3',1'){$a$}
\drawedge[curvedepth=1.5](4',3'){$b$}
\drawedge[curvedepth=1.6](7',4'){$a$}
\drawedge[curvedepth=4.5](7',1'){$c$}

\node[Nframe=n](Qn)(1,2){$Q_{8}$:}
\node(0)(2,0){$0$}
\node(1)(5.5,0){$1$}\rmark(1)
\node(2)(10,0){$2$}
\node(3)(14,0){$3$}
\node(4)(18,0){$4$}
\node(5)(22,0){$5$}
\node(6)(26,0){$6$}
\drawedge(0,1){$c$}
\drawedge(1,2){$a,b,c$}
\drawedge(2,3){$a,b$}
\drawedge(3,4){$a,c$}
\drawedge(4,5){$a,b$}
\drawedge(5,6){$a,b$}

\drawedge[curvedepth=1.5](3,1){$b$}
\drawedge[curvedepth=1.5](4,3){$c$}
\drawedge[curvedepth=1.6](6,4){$b$}
\drawedge[curvedepth=4.5](6,1){$a$}

\drawloop(5){$c$}
\drawloop(6){$c$}

\drawedge(1',0){$\varepsilon$}

\end{picture}\end{center}
\caption{The NFA $\cN$ for product $ L'_{9}(a,b,c) \cdot L_{8}(b,c,a)$. The empty states $8'$ and $7$ and the transitions to them are omitted.}
\label{fig:product3}
\end{figure}

We use the subset construction to determinize $\cN$ to get a DFA $\cP$ for the product. The states of $\cP$ are subsets of $Q_m'\cup Q_n$ and have one of three forms: 
$\{0'\}$, $\{1',0\}\cup  S$  or $\{ p'\}\cup  S$, where $p'=2',\dots,(m-1)'$ and   
$S\subseteq \{1,\dots, n-1\}$.

Note that for each $x\in \Sig$ every state $q\in Q_n \setminus \{0,n-1\}$  has a unique predecessor state $p\in Q_n\setminus \{n-1\}$  such that $px=q$.  For $w\in \Sig^*$, the $w$-predecessor of $S \subseteq Q_n \setminus \{0,n-1\}$  is denoted by $Sw^{-1}$.

\begin{lemma}
\label{lem:pred}
For each $n\ge 6$ and each $q\in Q_n$ there exists a word $w_q\in c\{a,b\}^*$ such that $1'w_q=3'$, $0w_q= q$, and each state of $Q_n\setminus \{0,q,n-1\}$ has a unique $w_q$-predecessor in $Q_n\setminus \{0,n-1\}$. In fact, the following $w_q$ satisfy these requirements:

\begin{equation}
	w_q = 
	\begin{cases}
		cab^2,  			&\text{if $q=1$;}\\
		ca,				&\text{if $q=2$;}\\
		cab^4, 			&\text{if $q=3$;}\\
		cab^2a^3b^{q-4},	&\text{if $q\ge 4$ and $q$ is even;}\\
		ca^4b^{q-5},		&\text{if $q\ge 5$ and $q$ is odd.}
	\end{cases}
\end{equation}
			
\end{lemma}
\begin{proof}
It is easily verified that in each case $1'w_q=3'$ and  $0w_q=q$.\
Note that $a$ and $b$ induce permutations on $Q_n\setminus \{0,n-1\}$ 
and $c$ is a one-to-one mapping from $Q_n\setminus \{0,n-1\}$ to $Q_n \setminus 
\{0,1\}$. Thus every state in $Q_n\setminus \{0,n-1\}$ that is mapped by $w_q$ to $Q_n\setminus \{0,n-1\}$ has a $w_q$-predecessor in $Q_n\setminus \{0,n-1\}$, and state $q$ has $0$ as its $w_q$-predecessor.  
\end{proof}

\begin{theorem}[Product: Ternary Case]
\label{prop:product}
For $m,n\ge 6$, the product $ L'_m(a,b,c)\cdot L_n(b,c,a)$ meets the bound $(m-1)2^{n-2}+1$.
\end{theorem}
\begin{proof}
Let $P$ consist of the following states
$\{0'\}$, $2^{n-2}$ sets of the form $\{1',0, n-1\}\cup S$ and $(m-2)2^{n-2}$   sets of the form $\{p',n-1\}\cup S$, 
where $S\subseteq \{1,\dots, n-2\}$, and $p'=2',\dots,(m-1)'$ -- a total of $(m-1)2^{n-2}+1$ sets.
We shall prove that all these states of $\cP$ are reachable and pairwise distinguishable. 
This together with the known upper bound will prove that the witnesses indeed meet the bound for product.

Consider the distinguishability of two states in $P$.
If we apply any word ending in $c$ to any subset of $Q_n\setminus \{0\}$, the resulting set does not contain the final state $1$.
Suppose that one of the states in our pair is $\{0'\}$; this is the only state accepting $bc$.
Next, if one of the states has the form $\{1',0,n-1\} \cup S$ and the other is $\{p',n-1\}\cup R$, then the former state accepts $c$, whereas the latter does not.
If $p < q$, then $\{p',n-1\}\cup R$ is distinguished from $\{q',n-1\}\cup S$ by $c^{m-q}$.
This leaves the case where the state in $Q'_m\setminus \{0'\}$ is the same in both sets in our pair. If the sets in $Q_n$ are $R$ and $S$, $R\neq S$, and $q\in R \oplus S$, then 
$a^{m-1-q}$ distinguishes these states.

Now we turn to reachability.
Since state $0'$ is initial in $\cN$, $\{0'\}$ is reachable. First we show that the sets in $P$ are reachable if $S=\emp$.
The set $\{1',0,n-1\}$ is reached by $ba^3$, 
$\{2',n-1\}$  by $ba$, $\{p',n-1\}$ by $bac^{p-2}$ for $p=3,\dots, (m-2)$, and $\{(m-1)',n-1\}$ by $b^3$.
\medskip

Now suppose all sets of the form
$\{1',0, n-1\}\cup S$ and $\{p',n-1\}\cup S$, $p=2,\dots,m-1$, with  $S\subseteq \{1,\dots, n-2\}$ and $|S|=k$, are reachable.
We show that if $k<n-2$, then every set with $S$ of size $k+1$ can be reached.
In each case we assume that $q\notin S$.
\be
\item
Sets with $3'$. We add $q$ to $S$ by applying $w_q$ to $(\{1', 0, n-1\}\cup Sw_q^{-1})$.
By Lemma~\ref{lem:pred}, every state except $q$ in $Q_n \setminus \{0,n-1\}$ has a unique $w_q$-predecessor in $Q_n \setminus \{0,n-1\}$.
		Hence, assuming we have $\{3', n-1\}\cup Sw_q^{-1}$ we can add $q$  since
	$$(\{1', 0, n-1\}\cup Sw_q^{-1}) w_q = \{3', n-1\}\cup S \cup \{q\}.$$ 
\item
Sets with $4'$:
	\be
	\item
	$(\{3', n-1\}\cup Sb^{-1} \cup \{3\}) b = \{ 4', n-1\}\cup S \cup \{1\}$,
	\item
	$(\{3', n-1\}\cup Sb^{-1} \cup \{1\}) b = \{ 4', n-1\}\cup S \cup \{2\}$, 
	\item
	$(\{3', n-1\}\cup Sb^{-1} \cup \{2\}) b = \{4', n-1\}\cup S \cup \{3\}$, 
	\item
	$(\{3', n-1\}\cup Sb^{-1} \cup \{m-2\}) b = \{4', n-1\}\cup S \cup \{4\}$, 
	\item
	$(\{3', n-1\}\cup Sb^{-1} \cup \{q\}) b = \{4', n-1\}\cup S \cup \{q+1\}$, for $q=4,\dots, 	n-2$.
	\ee
\item
Sets with $p'$, $p=5,\dots, m-2$:
	$$(\{(p-1)', n-1\}\cup S(b^2ab)^{-1} \cup \{3\}) b^2ab = \{p', n-1\}\cup S \cup \{1\},$$ 
	$$(\{(p-1)', n-1\}\cup Sa^{-1} \cup \{q-1\}) a = \{p', n-1\} \cup S \cup \{q\} \text{ for $q=2,\dots, n-2$.}$$
\item
Sets with $(m-1)'$:
If $1\notin S$, then $S$ has a $c$-predecessor in $Q_n\setminus \{0,n-1\}$ and $(\{1',0, n-1\}\cup Sc^{-1})c = \{2', n-1\}\cup S \cup \{1\}$. Thus we have 

	$$(\{2', n-1\}\cup S(bab)^{-1} \cup \{1\}) bab = \{(m-1)', n-1\}\cup S \cup \{1\},$$
	$$(\{2', n-1\}\cup Sb^{-1} \cup \{1\}) b = \{(m-1)', n-1\}\cup S \cup \{2\},$$ 
	$$(\{(m-1)', n-1\}\cup Sa^{-1} \cup \{q-1\}) a = \{(m-1)', n-1\}\cup S \cup \{q\}, $$
	for $q=3,\dots, n-2$.

\item
Sets with $1'$:
	$$(\{3', n-1\}\cup Sa^{-1} \cup \{n-2\}) a = \{1',0, n-1\}\cup S \cup \{1\},$$ 
	$$(\{3', n-1\}\cup Sa^{-1} \cup \{1\}) a = \{1',0, n-1\}\cup S \cup \{2\},$$ 
	$$(\{2', n-1\}\cup S(a^2)^{-1} \cup \{1\}) a^2 = \{1',0, n-1\}\cup S \cup \{3\},$$
	$$(\{1',0, n-1\}\cup S(a^3)^{-1} \cup \{q-3\}) a^3 = \{1',0, n-1\}\cup S \cup \{q\}, $$
	for $q=4,\dots, n-2$.
\item
Sets with $2'$:
	$$(\{1', 0,n-1\}\cup Sc^{-1}) c = \{2',n-1\}\cup S \cup \{1\}, \text{ as in Case~4},$$ 
	$$(\{1', 0,n-1\}\cup Sa^{-1} \cup \{q-1\}) a = \{2',n-1\}\cup S \cup \{q\},  $$			for $q=2,\dots, n-2.$ 
\ee
Thus the induction step goes through and all required sets are reachable.
\end{proof}

Cmorik and Jir\'askov\'a~\cite[Theorem 5]{CmJi12} also found binary witnesses that meet the bound $(m-1)2^{n-2}$ in the case where $m-2$ and $n-2$ are relatively prime.
It remained unknown whether the bound $(m-1)2^{n-2}+1$ is reachable with a binary alphabet.
We show in the Appendix that a slight modification of the first witness of~\cite{CmJi12} meets the upper bound exactly.
 
For $m\ge 6,n\ge 3$, let the first DFA be that of~\cite{CmJi12}, except that the set of final states is changed to $\{2',4'\}$; thus let $\Sig=\{a,b\}$,
$\cD'_m(a,b)=(Q'_m,\Sig,\delta',0',\{2',4'\})$, and
let $\cD_n(a,b)=(Q_n,\Sig,\delta,0,\{1\})$, with the transitions shown in Figure~\ref{fig:product2}.
Let $L'_m(a,b)$ and $L_n(a,b)$ be the corresponding languages.

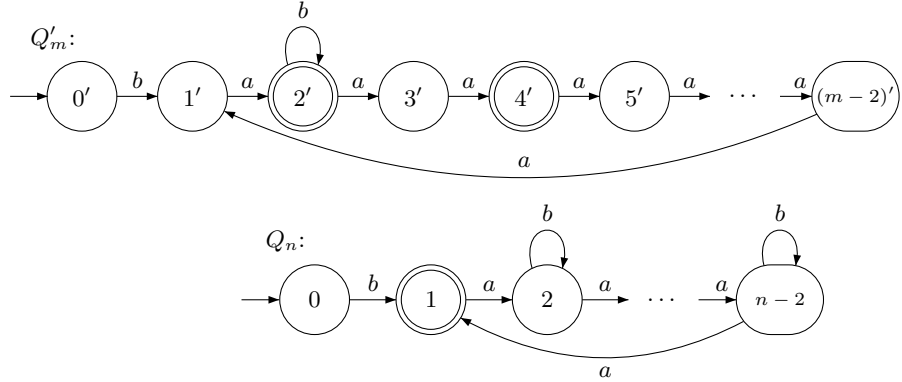
\begin{figure}[ht]
\unitlength 11pt
\begin{center}\begin{picture}(31,11)(0,0)
\gasset{Nh=2.4,Nw=2.4,Nmr=1.2,ELdist=0.3,loopdiam=1.2}
\node[Nframe=n](Qm)(1,10){$Q'_m$:}
\node(0')(2,8){$0'$}\imark(0')
\node(1')(5.8,8){$1'$}
\node(2')(9.6,8){$2'$}\rmark(2')
\node(3')(13.4,8){$3'$}
\node(4')(17.2,8){$4'$}\rmark(4')
\node(5')(21,8){$5'$}
\node[Nframe=n](qdots)(24.8,8){$\dots$}
\node[Nw=3](m-2')(28.6,8){{\scriptsize $(m-2)'$}}
\drawedge(0',1'){$b$}
\drawedge(1',2'){$a$}
\drawedge(2',3'){$a$}
\drawedge(3',4'){$a$}
\drawedge(4',5'){$a$}
\drawedge(5',qdots){$a$}
\drawedge(qdots,m-2'){$a$}
\drawedge[curvedepth=2.8,ELdist=-0.7](m-2',1'){$a$}
\drawloop(2'){$b$}
\node[Nframe=n](Qn)(9,3){$Q_n$:}
\node(0)(10,1){$0$}\imark(0)
\node(1)(14,1){$1$}\rmark(1)
\node(2)(18,1){$2$}
\node[Nframe=n](ndots)(22,1){$\dots$}
\node[Nw=3](n-2)(26,1){\scriptsize$n-2$}
\drawedge(0,1){$b$}
\drawedge(1,2){$a$}
\drawedge(2,ndots){$a$}
\drawedge(ndots,n-2){$a$}
\drawedge[curvedepth=2](n-2,1){$a$}
\drawloop(2){$b$}
\drawloop(n-2){$b$}
\end{picture}\end{center}
\caption{The witness DFA $\cD_m'$ and $\cD_n$ for product. The transitions to the empty states $(m-1)'$ and $n-1$ are omitted.}
\label{fig:product2}
\end{figure}

\begin{theorem}[Product: Binary Case]
\label{thm:product_binary}
For $m,n \ge 6$, $L'_m(a,b)$ is suffix-free and $L'_m(a,b)\cdot L_n(a,b)$ meets the bound $(m-1)2^{n-2}+1$ when $m-2$ and $n-2$ are relatively prime.
\end{theorem}
\begin{proof}
First we need to show that $L'_m(a,b)$ is minimal and suffix-free.
For minimality, it is easy to verify that every pair of states is distinguished by a word of the form $a^i b$.
Suppose that $L'_m(a,b)$ is not suffix-free; then there are some words $u,v$ such that $v \in L'_m(a,b)$ and $uv \in L'_m(a,b)$.
Since there is no transformation mapping $0'$ to itself, this means that $v$ maps both $0'$ and $q' = 0' u \in \{1',\ldots,(m-2)'\}$ to a final state.
Clearly $v=bw'$, for some $w$ and so $q$ must be $2'$. Then $w$ maps $1'$ and $2'$ to a final state. Since these two states cannot be merged to any state other than $(m-1)'$, and $b$ sends every state from $\{1',\ldots,(m-2)'\}$ except $2'$ to $(m-1)'$, we have $w = a^i$. But $a^i$ preserves the distance between the states in the cycle, and so $1'$ and $2'$ cannot be mapped simultaneously  to $2'$ and $4'$.

The proof for meeting the bound $(m-1)2^{n-2}+1$ is similar to that of~\cite[Theorem 5]{CmJi12}.
We construct the NFA for $L'_m(a,b) L_n(a,b)$ as usual.
We show that $(m-1)2^{n-2}+1$ subsets of the states are reachable and distinguishable.
Since a subset $S$ without $n-1$ cannot be distinguished from $S \cup \{n-1\}$, we consider only $S$ as the representative for both $S$ and $S \cup \{n-1\}$.
Similarly, a reachable subset $S$ contains 0 if and only if it contains either $2'$ or $4'$; we use only $S \setminus \{0\}$ as the representative for these subsets.

Let $X \subseteq \{1,\ldots,n-2\}$.
If $X = \emptyset$, then clearly each of the $m-1$ subsets $\{p'\}$ for $p \in \{1,\ldots,n-1\}$ is reachable by $ba^{p-1}$.
Let $X \neq \emptyset$.
To show that each of the $(m-2)(2^{n-2}-1)$ subsets $S = \{p'\} \cup X$ for $i \in \{1,\ldots,q_{m-2}\}$ is reachable, we follow  the proof of~\cite[Theorem 5]{CmJi12} exactly. In this proof note that $4' \in F$ makes no difference, since $b$ maps $4'$ to $(m-1)'$, and so if applied to a subset containing $4'$, it results in a set with 
$(m-1)'$.
For the $2^{n-2}-1$ subsets $S = \{(m-1)'\} \cup X$ consider $S' = \{1'\} \cup X$. Then $S = S'b$, and since $S'$ is reachable, so is $S$.

Finally, we show that all these $(m-1)2^{n-2}$ sets together with the initial state $\{0'\}$ are distinguishable.
Set $\{0'\}$ is distinguished from every other subset by $bab$.
Consider $\{p'\} \cup X$ and $\{q'\} \cup Y$ for distinct $p,q \in \{1,\ldots,m-2\}$, and some $X,Y \subseteq \{1,\ldots,n-2\}$.
Then $w={a^{m- p}b}$ maps $p'$ to $2'$, and $q'$ to $(m-1)'$. Thus $(\{p'\} \cup X)wb$ contains the final state $1$, and $(\{q'\} \cup X)wb$ does not.
Consider $\{p'\} \cup X$ and $\{q'\} \cup Y$ with $X \neq Y$; then $X$ and $Y$ differ in some $r \in \{1,\ldots,n-2\}$.
Thus $a^{n-r-1}$ distinguishes these subsets.
\end{proof} 

\begin{theorem}
Suppose $m,n \ge 4$ and $L'_m L_n$ meets the bound $2^{n-2}+1$. Then
the transition semigroup $T_n$ of a minimal DFA $\cD_n$ of $L_n$ is a subsemigroup of $\Vsf(n)$ and is not a subsemigroup of $\Wsf(n)$.
\end{theorem}
\begin{proof}
Let $\cD'_m = (Q'_m, \Sig, \delta', 0', F')$ with $Q'_m=\{0',\ldots,(m-1)'\}$, and let $\cD_n = (Q_n, \Sig, \delta, 0, F)$ with $Q_n = \{0,\ldots,n-1\}$.
We construct the NFA for $L'_m L_n$ as usual.

From~\cite[Lemma 9]{HaSa09} we know that the set of reachable and distinguishable subsets of
$Q'_m \cup Q_n$ can be represented by:
$\{p'\} \cup X$ for each $p \in \{1,\ldots,m-1\}$ and $X \subseteq \{1,\ldots,n-2\}$, or by $\{0'\}$.
A reachable subset $S$ contains $0$ if and only if $S \cap F' \neq \emptyset$. 
Also $S$ is not distinguishable from $S \cup \{(m-1)'\}$ and $S \cup \{n-1\}$.
To reach the bound $(m-1)2^{n-2}+1$, all these subsets or their equivalents must be reachable.

Suppose that $p,q \in \{1,\ldots,n-2\}$ are not colliding in $Q_n$.
Consider $S$ that contains both $p$ and $q$. Then $S$ is reached from some $S'$ by a transformation $t$, where $S'$ contains $0$ and a state $r \in \{1,\ldots,n-2\}$ such that $0t=p$ and $rt=q$ (or $0t=q$ and $rt=p$). But since $p$ and $q$ are not colliding, there is no such transformation in $T_n$.
Thus all pairs in $Q_n \setminus \{0,n-1\}$ must be colliding and $T_n$ is a subsemigroup of $\Vsf$.
\end{proof}

\subsection{Boolean Operations}

Let $\cD = \cD(a,b,c)$ be the DFA of Definition~\ref{def:D5}, and let $L = L(a,b,c)$ be the language of this DFA. 
Consider the partial permutation $\pi'(a) = a$, $\pi'(b) = b$, and $\pi'(c) = -$. The dialect associated with $\pi'$ is $\cD'(a,b,-)$, which is $\cD(a,b,c)$ restricted to the alphabet $\{a,b\}$.
This is our first witness for boolean operations, and we prime its states to distinguish them from those of the second witness defined below. This DFA is illustrated in Figure~\ref{fig:boolean} for $m=9$.

Now take the partial permutation $\pi$, where $\pi(a) = -$, $\pi(b) = b$, and $\pi(c) = a$;
here $a$ plays the role of $c$. 
Thus $\cD(-,b,a)$ is the DFA in which  
$a \colon (0\to n-1)(1,\dots,n-2)$ and $b \colon (2\to n-1)(1\to 2)(0 \to 1)(3,4)$, as illustrated in Figure~\ref{fig:boolean} for $n=8$.
This DFA is our second witness for boolean operations.
The language of $\cD(-,b,a)$ is the dialect $L(-,b,a)$ of $L$.

\begin{figure}[t]
\unitlength 12pt
\begin{center}\begin{picture}(31,12)(0,-4)
\gasset{Nh=2.,Nw=2,Nmr=1.5,ELdist=0.3,loopdiam=1.0}
\node[Nframe=n](Qm)(1,8){$Q'_{ 9}$:}
\node(0')(2,6){$0'$}\imark(0')
\node(1')(5.5,6){$1'$}\rmark(1')
\node(2')(10,6){$2'$}
\node(3')(14,6){$3'$}
\node(4')(18,6){$4'$}
\node(5')(22,6){$5'$}
\node(6')(26,6){$6'$}
\node(7')(30,6){$7'$}
\drawloop(5'){$b$}
\drawloop(6'){$b$}
\drawloop(7'){$b$}
\drawedge(0',1'){$b$}
\drawedge(1',2'){$a,b$}
\drawedge(2',3'){$a$}
\drawedge(3',4'){$b$}
\drawedge(4',5'){$a$}
\drawedge(5',6'){$a$}
\drawedge(6',7'){$a$}

\drawedge[curvedepth=1.5](3',1'){$a$}
\drawedge[curvedepth=1.5](4',3'){$b$}
\drawedge[curvedepth=1.6](7',4'){$a$}

\node[Nframe=n](Qn)(1,2){$Q_{ 8}$:}
\node(0)(2,0){$0$}\imark(0)
\node(1)(5.5,0){$1$}\rmark(1)
\node(2)(10,0){$2$}
\node(3)(14,0){$3$}
\node(4)(18,0){$4$}
\node(5)(22,0){$5$}
\node(6)(26,0){$6$}
\drawedge(0,1){$b$}
\drawedge(1,2){$a,b$}
\drawedge(2,3){$a$}
\drawedge(3,4){$a,b$}
\drawedge(4,5){$a$}
\drawedge(5,6){$a$}

\drawedge[curvedepth=1.5](4,3){$b$}
\drawedge[curvedepth=3.0](6,1){$a$}

\drawloop(5){$b$}
\drawloop(6){$b$}

\end{picture}\end{center}
\caption{Witnesses $\cD'_9(a,b,-)$ and $\cD_8(-,b,a)$ for boolean operations. The empty states $8'$ and $7$ and the transitions to them are omitted.}
\label{fig:boolean}
\end{figure}

\begin{theorem}
\label{thm:boolean}
For $m,n \ge 6$, $L'_m(a,b,-)$ and $L_n(-,b,a)$ meet the bounds $mn - (m+n-2)$ for union and symmetric difference, $mn - 2(m+n-3)$ for intersection, and $mn - (m+2n-4)$ for difference.
\end{theorem}
\begin{proof}
As usual, we construct the direct product $\cP$ of $\cD'_m$ and $\cD_n$.
The initial state of $\cP$ is $(0',0)$ which reaches only $(1',1)$ and can never be re-entered. Thus $0'$ and $0$ do not appear in any reachable  state of $\cP$ other than $(0',0)$. We show that the remaining $(m-1)(n-1)$ states are all reachable, for a total of $1+(m-1)(n-1)= mn-(m+n-2)$ states. Let $m$ define the row and $n$ the column in the direct product.

As a preliminary step we show how $((m-2)',n-2)$ can be reached. 
First we reach $(4',4)$ by $ba^2b$.
If $m=6$, then $(4',4) a^{q-4} = (4', q)$ for $q=5,\dots, (n-2)$.
Otherwise, if $m> n$, then $(0',0) ba^3 (ba)^2= (5',4)$, $(5',4) (ba)^{m-n-1} = ((m-n+4)',4)$, and
$(m-n+4,4) a^{n-6} = ((m-2)',n-2)$. 
If $m=n$, then $(4',4)a^{n-6}= ((m-2)',n-2)$.
Let $m<n$.
Observe that if $(n-m) \equiv p \pmod 3$ for $0 \le p \le 2$, then
$((3-p)',4) a^{n-m} b a^{m-6} = ((m-2)',n-2)$.
Hence we need to reach $((3-p)',4)$.
We consider the following three cases:
\be
\item $(n-m) \equiv 0 \pmod{3}$.

We reach $(3',4)$ by $b^2a^2ba$, and then apply $a^{n-m} b a^{m-6}$.

\item $(n-m) \equiv 2 \pmod{3}$.

We reach $(1',4)$ by $b^2a^2$, and then apply $a^{n-m} b a^{m-6}$.

\item $(n-m) \equiv 1 \pmod{3}$.

If $m \le n-4$ then $(1',4) a^3 b a^{n-m-3} b a^{m-6} = ((m-2)',n-2)$.
Otherwise $m=n-1$.
If $n < 12$, then we have five cases, where $(0',0)ba^2=(3',3)$:
\be
\item $m=6$, $n=7$: $(4',4)a^2ba^2 =(2',4)$.
\item
$m=7$, $n=8$: $(3',3) (ab)^2 a^4 ba^2 =(2',4)$.
\item
$m=8$, $n=9$: $(3',3) a^8 =(2',4)$.
\item
$m=9$, $n=10$: $(3',3) aba^9 =(2',4)$.
\item
$m=10$, $n=11$: $(3',3) ab a^4b a^3b ab a^6 ba^2 =(2',4)$.
\ee

Otherwise we reach $(2',4)$ as follows:
$(4',4) a^{m-6} = (4',n-2)$,
$(4',n-2) a^{m-6} = (4',n-7)$,
$(4',n-7) bab a^2b a^3 = (2',1)$,
$(2',1) a^3 = (2',4)$.
\ee
Now, having reached $((m-2)',n-2)$, we can reach all the remaining pairs: We have
$((m-2)',n-2) (ab)^2 b = ((m-1)', 3)$, and from $((m-1)', 3)$
all pairs $((m-1)', q)$ for $1 \le q \le n-1$ can be reached, since $Q_n \setminus \{0,n-1\}$ is a strongly connected component in $\cD_n$.
Similarly, $((m-2)',n-2) a b^2 = (4',n-1)$, and all pairs $(p',n-1)$ for $1 \le p \le m-1$ can be reached.
For the remaining pairs we proceed as follows:

	\be
	\item Column 3: 
		\be 
		\item
		$((m-2)',n-2) aba = (1',3)$, 
		\item
		$(0,0) ba^3b = (2',3)$,
		\item
		$(0,0) ba^2 = (3',3)$, 
		\item
		$(2',3) ab = (4',3)$,
		\item
		$((p-1)',3) ab = (p',3)$ for $p=5,\dots,m-2$. 
		\ee
	\item Column 4: 
		\be 
		\item
		$(3',3) a = (1',4)$,
		\item
		$(3',3) b = (4',4)$, 
		\item
		$((p-1)',3) a = (p',4)$ for $p=2,3,5,\dots,m-2$. 
		\ee
	\item Column $q$ for $q=5,\dots, n-2$: 
		\be 
		\item
		$(3',q-1) a = (1',q)$, 
		\item
		$(3',q) b = (4',q)$,
		\item
		$((p-1)',q-1) a = (p',q)$ for $p=2,3,5,\dots,m-2$.
		\ee
 	\item Column $1$: 
		\be 
		\item
		$(0',0) b = (1',1)$,
		\item
		$((m-2)',n-2) a = (4',1)$.
		\item
		$((p-1)',n-2) a = (p',1)$ for $p=2,3,5,\dots,m-2$.
		\ee
	\item Column $2$: 
		\be 
		\item
		$(3',1) a = (1',2)$,
		\item
		$(0',0) a^2 = (2',2)$,
		\item
		$(3',1) b = (4',2)$,
		\item
		$((p-1)',1) a = (p',2)$ for $p=3,5,\dots,m-2$. 
		\ee
	\ee
	
As we discussed in Subsection~\ref{subsec:product}, for each $x\in \Sig$ every state $q\in Q_n \setminus \{0,n-1\}$ has a unique predecessor state $p\in Q_n\setminus \{n-1\}$. It follows that if $q w = p \in Q_n \setminus \{0,n-1\}$, for some state $q$ and word $w$, then $r w \neq p$ for $r \neq q$.
The same facts apply to $Q'_m$.
We shall need the   following two claims:
\medskip

\noin
\textit{Claim 1.} From any pair $(q',p)$ with $1 \le q \le m-2$, $1 \le p \le n-2$ we can reach $(1',1)$.

First we find a word $w$ such that $q'w \in \{1',2',3'\}$ and $pw = 1$.
Note that it is sufficient to find $u$ such that $q'u\in \{1',2',3'\}$ and $pu \in Q_n \setminus \{0,n-1\}$,
because then $w = u a^{n-1-(pu)}$ does the job.
So if $q' \in \{1',2',3'\}$, then we are done.
Otherwise use $a^{m-1-q}$, which maps $q'$ to $4'$.
If $p a^{m-1-q} \neq 2$, then we use $b$ to map $4'$ to $3'$ and keep the state of 
$\cD_n$ in $Q_n\setminus \{0,n-1\}$, and again we are done.
Suppose $p a^{m-1-q} = 2$. If $m=6$ then $4' ab = 3'$ and $p a^{m-1-q} ab \neq 2$, and we are done.
Otherwise $4' a^{m-5} = 4'$, and also $4' a b a^{m-6} = 4'$.
If $2 a^{m-5} \neq 2$ then we can apply an additional $b$ and be done.
If $2 a^{m-5} = 2$, then $2 a b a^{m-6} = 3$, and again by applying an additional $b$ we are done.

Now let $q' \in \{2',3'\}$. We show that there is a word $w$ such that $q' w = 1$ and $1 w = 1$.
If $q' a^{n-2} = 1'$ then we are done, and if $q' a^{n-2} \in \{2',3'\} \setminus \{q'\}$, then $q' (a^{n-2})^2 = 1$ and we have reached $(1',1)$.
So assume $q' a^{n-2} = q'$; thus $(n-2) \equiv 0 \mod{3}$, and so $n \ge 8$.

If $q'=3'$ then $(3',1) a^4 b a^{n-1-4} = (r',1)$ with $r' \neq 3'$, so we are done by $(a^4 b a^{n-1-4})^2$.
If $q'=2'$ then $(2',1) a^5 b a^{n-1-5} = (r',1)$ with $r' \neq 2'$, so we are done by $a^5 b (a^{n-1-5})^2$.

\noin
\textit{Claim 2.}
For any $q' \in Q'_m \setminus \{0',(m-1)'\}$ and $p \in Q_n \setminus \{0,(n-1)\}$ there is a word such that $q'w \in Q'_m \setminus \{0',(m-1)'\}$ and $pw = n-1$.
 It follows that $(q',p)$ can be mapped to $(r',n-1)$, for any state $r' \in Q'_m \setminus \{0'\}$.

If $q' \in \{1',2',3'\}$ then let $w_1 = a^2b$, $ab$, or $b$, respectively, and $\eps$ otherwise; so $w_1$ maps $q'$ to $\{4',\ldots,(m-2)'\}$.
Then let $w_2 = a^{n - (p w_1)}$; so $w_2$ maps $p w_1$ to $2$, and $q w_1$ is still in $\{4',\ldots,(m-2)'\}$.
Then $w = w_1 w_2 b$ satisfies the claim.
Since $Q_m'\setminus \{0',(m-1)'\}$ is strongly connected, we can map $(q'w,n-1)$ to any $r' \in Q'_m \setminus \{0'\}$.

Now we consider the four cases of the operations.
In all cases pair $(0',0)$ is distinguished as the only non-empty state which does not accept any word starting with $a$.

\be
\item Union and symmetric difference.

Consider $(q'_1,p_1)$ and $(q'_2,p_2)$ with $q'_1 \neq q'_2$, where $1 \le q_1,q_2 \le m-1$ and $1 \le p_1,p_2 \le n-1$.
Without loss of generality, $q'_1 \neq (m-1)'$.
The same arguments here work for both operations, since we are not mapping the pairs to $(1',1)$.
By Claim~2 we can map $(q'_1,p_1)$ to $(1,n-1)$ by some word $w$. Then $q'_2 w \neq 1$.
If also $p_2 \neq 1$, then $(q'_2,p_2)$ is not  final  and $w$ distinguishes our pairs.
Otherwise we apply Claim~2 once more for $(1,p_2)$, obtaining a word $u$ such that $(1,n-1)u = (1,n-1)$, and $(q'_2,1)u = (q'_2 u,n-1)$.
Since $1u = u$, we have $q'_2 u \neq 1$, and so $w' u$ distinguishes our pairs.

Assume now $q'_1 = q'_2$ and $p_1 \neq p_2$, where $p_1 \neq n-1$. Let $w$ be a word mapping $p_1$ to $1$; then $p_2 w \neq 1$.
If $q'_1 w \neq 1'$, $(q'_2,p_2) w = (q'_1,p_2) w$ is not  final  so $w$ distinguishes our pairs.
Otherwise $w a^3 b^2$ maps $q'_1$ to $(m-1)'$, and $p_1$ to $4$.
Thus $w a^3 b^2 a^{n-6}$ maps $p_1$ to $1$, and since $p_2$ is mapped elsewhere and $q'_1 = q'_2$ to $(m-1)'$, our pairs are distinguished.

\item Intersection.

Consider $(q'_1,p_1)$ and $(q'_2,p_2)$ with $q'_1 \neq q'_2$ or $p_1 \neq p_2$, where $1 \le q_1,q_2 \le m-2$ and $1 \le p_1,p_2 \le n-2$.
By Claim~1 we can map $(q'_1,p_1)$ to the  final state $(1',1)$ by some word $w$.
Then either $q'_2 w \neq 1'$ or $p_2 w \neq 1$, $(q'_2,p_2)$ is not  final  and our pairs are distinguished.
Together with $(0',0)$ and $((m-1)',n-1)$ these give $2+(m-2)(n-2) = mn - 2(m+n-3)$ distinguished pairs.

\item Difference.

Consider $(q'_1,p_1)$ and $(q'_2,p_2)$ with $q'_1 \neq q'_2$, where $1 \le q_1,q_2 \le m-2$ and $1 \le p_1,p_2 \le n-1$.
This follows  in  exactly the same  way as the corresponding case of union and symmetric difference.
Assume now $q'_1 = q'_2$ and $p_1 \neq p_2$.
Without loss of generality, $p_1 \neq n-1$.
By Claim~1 we can map $(q'_1,p_1)$ to the non- final state $(1',1)$ by some word $w$.
Then $(q'_2,p_2)w = (1',p_2w)$, and since $p_2 w \neq p_1 w = 1$, $(1',p_2w)$  a final state.
Thus $w$ distinguishes our pairs.

Together with $(0',0)$ and $((m-1)',n-1)$ these give $2+(m-2)(n-1) = mn-(m+2n-4)$ distinguished pairs.
\ee
\end{proof}

\section{Witnesses with Semigroups in $\Wsf(n)$}

We now turn to the operations which cannot have witnesses with transition semigroups in $\Vsf$.

\begin{definition}
\label{def:D6}
For $n\ge 4$, 
we define the DFA 
$\cD_n(a,b,c,d,e) =(Q_n,\Sig,\delta,0,F),$
where $Q_n=\{0,\ldots,n-1\}$, $\Sig=\{a,b,c,d,e\}$, 
 $\delta$ is defined by the transformations
\bi
\item
$a\colon (0\to n-1) (1,\ldots,n-2)$,
\item
$b\colon (0\to n-1) (1,2)$,
\item
$c\colon (0\to n-1) (n-2\to 1)$,
\item
$d\colon (\{0,1\}\to n-1)$,
\item
$e\colon (Q\setminus \{0\} \to n-1)(0\to 1)$,
\ei
and $F=\{q\in Q_n\setminus \{0,n-1\} \mid q \text{ is odd}\}$.
For $n=4$, $a$ and $b$ coincide, and we can use $\Sig=\{b,c,d,e\}$.
The structure of $\cD_5(a,b,c,d,e)$ is illustrated in Figure~\ref{fig:witness}.
\end{definition}

Our main result in this section is the following theorem:
\begin{theorem}[Boolean Operations, Reversal,  Number and  Complexity of Atoms, Syntactic Complexity]
\label{thm:witness6}
Let $\cD_n(a,b,c,d,e)$ be the DFA of Definition~\ref{def:D6}, and let the language it accepts be
$L_n(a,b,c,d,e)$. 
Then $L_n(a,b,c,d,e)$ meets the bounds for  boolean operations, reversal, number  and quotient complexity  of atoms, and syntactic complexity as follows:
\be
\item
For $n,m \ge 4$, $L_m(a,b,-,d,e)$ and $L_n(b,a,-,d,e)$ meet the bounds $mn - (m+n-2)$  for union and symmetric difference, $mn - 2(m+n-3)$ for intersection, and $mn - (m+2n-4)$ for difference.
\item
For $n\ge 4$, $L_n(a,-,c,-,e)$ meets the bound $2^{n-2}+1$ for reversal and number of atoms.
\item
For $n\ge 6$, $L_m(a,b,c,d,e)$ meets the bound $(n-1)^{n-2}+n-2$ for syntactic complexity, and the bounds on the quotient complexities of atoms.
\ee
\end{theorem}

The claim about syntactic complexity is known from~\cite{BrSz15}.
It was shown in~\cite{BrTa14} that 
the number of atoms of a regular language $L$ is equal to the quotient complexity of $L^R$.
In the next subsections we prove the claim about boolean operations, reversal, and atom complexity.
First we state some properties of $\cD_n$.

\begin{proposition}
\label{prop:sf_wit}
For $n\ge 4$ the DFA of Definition~\ref{def:D6} is minimal, accepts a suffix-free language, and its transition semigroup $T_n$ has cardinality $(n-1)^{n-2}+n-2$. In particular, $T_n$ contains 
(a) all $(n-1)^{n-2}$ transformations that send $0$ and $n-1$ to $n-1$ and map 
$Q\setminus \{0,n-1\}$ to $Q\setminus \{0\}$, and 
(b) all $n-2$ transformations that send $0$ to a state in $Q\setminus \{0,n-1\}$ and map all the other states to $n-1$. Also, $T_n$ is generated by $\{a,b,c,d,e\}$ and cannot be generated by a smaller set of transformations.
\end{proposition}
\begin{proof}
To prove minimality, note first that only the initial state $0$ accepts $e$, and only state $n-1$ accepts nothing.
Suppose $p,q\in Q_n\setminus\{0, n-1\}$, $p<q$, and $p$ and $q$ have the same parity; then $pa^{n-2-p}c=1$, which is a final state, but $pa^{n-2-q}c$ is an even state, which is non-final. If $p$ and $q$ have different parities, then one is final while the other is not.
Hence $\cD_n$ is minimal.

The language $L_n$ is suffix-free, because every word in $L_n$ has the form $ex$ with $x\in \{a,b,c,d\}$.
The claims about cardinality of $T_n$, transformations, and generators were proved in~\cite{BLY12,BrSz15}.
\end{proof}

\subsection{Boolean Operations}

We now show that witness DFAs for boolean operations may have transition semigroups in $\Wsf$. Thus boolean operations have witnesses with transition semigroups that are subsemigroups of either $\Vsf$ or $\Wsf$. This is sharp contrast with all the other complexity measures: any witness for star and any second witness for product can only be associated with $\Vsf$, while the size of the syntactic semigroup, reversal, and complexities of atoms must have witnesses associated with $\Wsf$.

\begin{theorem}
\label{thm:boolean2}
For $n,m \ge 4$, $L_m(a,b,-,d,e)$ and $L_n(b,a,-,d,e)$ meet the bounds $mn - (m+n-2)$ for union and symmetric difference, $mn - 2(m+n-3)$ for intersection, and $mn - (m+2n-4)$ for difference.
\end{theorem}
\begin{proof}
Our two DFAs $\cD_m(a,b,-,d,e)$ and $\cD_n(b,a,-,d,e)$ are illustrated in Figure~\ref{fig:boolean2}; consider their direct product.
Let $ H=\{((m-1)',q) \mid 1 \le q \le n-2\}$, 
$ V=\{(p',n-1) \mid 1 \le p \le m-2\}$, and 
$M=\{(p',q) \mid 1 \le p \le m-2, 1 \le q \le n-2\}$.

Let $\cA_n(a,b)=(Q_n\setminus \{0,n-1\},\{a,b\},\delta,1,F)$ be the DFA obtained from the DFA $\cD_n(a,b,-,d,e)$ by restricting the alphabet to $\{a,b\}$ and the set of states to $\{1,\dots,n-2\}$.
Since DFAs $\cA'_m(a,b)$ and $\cA_n(b,a)$ have ordered pairs $\{\delta'_a,\delta'_b\}$ and $\{\delta_a,\delta_b\}$ of transformations  that generate the symmetric groups of degrees $m$ and $n$ and are not conjugate, the result from~\cite[Theorem 1]{BBMR14} applies, except in our case where $m=4$ and $n=4$ (we add two states to the $m$ and $n$ in~\cite{BBMR14}).
We have verified this case by computation.
Therefore we know that all states in $M$ are reachable
from state $(1',1)$ and all pairs of such states are distinguishable.

\begin{figure}[thb]
\unitlength 10pt
\begin{center}\begin{picture}(28,15)(0,-3)
\gasset{Nh=2.,Nw=2,Nmr=1.5,ELdist=0.3,loopdiam=1.0}
\node[Nframe=n](Qm)(0,10.5){$Q'_8(a,b,-,d,e)$:}
\node(0')(2,8){$0'$}\imark(0')
\node(1')(5.5,8){$1'$}\rmark(1')
\node(2')(10,8){$2'$}
\node(3')(14,8){$3'$}\rmark(3')
\node(4')(18,8){$4'$}
\node(5')(22,8){$5'$}\rmark(5')
\node(6')(26,8){$6'$}
\drawloop(2'){$d$}
\drawloop(3'){$b,d$}
\drawloop(4'){$b,d$}
\drawloop(5'){$b,d$}
\drawloop(6'){$b,d$}
\drawedge(0',1'){$e$}
\drawedge(1',2'){$a,b$}
\drawedge(2',3'){$a$}
\drawedge(3',4'){$a$}
\drawedge(4',5'){$a$}
\drawedge(5',6'){$a$}

\drawedge[curvedepth=-2.5,ELside=r](2',1'){$b$}
\drawedge[curvedepth=2.5](6',1'){$a$}

\node[Nframe=n](Qn)(0,2.5){$Q_7(b,a,-,d,e)$:}
\node(0)(2,0){$0$}\imark(0)
\node(1)(5.5,0){$1$}\rmark(1)
\node(2)(10,0){$2$}
\node(3)(14,0){$3$}\rmark(3)
\node(4)(18,0){$4$}
\node(5)(22,0){$5$}\rmark(5)
\drawedge(0,1){$e$}
\drawedge(1,2){$a,b$}
\drawedge(2,3){$b$}
\drawedge(3,4){$b$}
\drawedge(4,5){$b$}

\drawedge[curvedepth=-2.5,ELside=r](2,1){$a$}
\drawedge[curvedepth=2.5](5,1){$b$}

\drawloop(2){$d$}
\drawloop(3){$a,d$}
\drawloop(4){$a,d$}
\drawloop(5){$a,d$}

\end{picture}\end{center}
\caption{The DFAs $\cD'_8(a,b,-,d,e)$ and $\cD_7(b,a,-,d,e)$ for boolean operations. The empty states $8'$ and $7$ and the transitions to them are omitted.}
\label{fig:boolean2}
\end{figure}
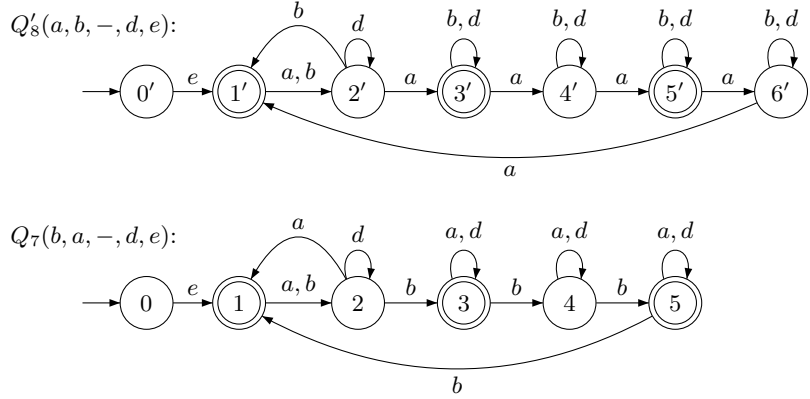

We show that $(0',0)$, $((m-1)',n-1)$ and all states in $H$, $V$ and $M$ are reachable.
State $(0',0)$ is the initial state, $(1',1)$ is reached from $(0',0)$ by $e$ and 
$((m-1)',n-1)$ by $a$.
By Theorem~1 of~\cite{BBMR14} all the states in $M$, and in particular $(1',2)$ and $(2',1)$, are reachable from $(1',1)$.
From state $(1',2)$ we reach $((m-1)',2)$ by $d$ and from there, state $((m-1)',q)$ by a word in $b^*$ for $1 \le q \le n-2$.
Symmetrically, from $(2',1)$ we reach $(2',n-1)$, and from there, state $(p',n-1)$ for $1 \le p \le m-2$ by a word in $a^*$.
Hence all $(m-1)(n-1)+1=mn-(m+n-2)$ states are reachable.

State $(0',0)$ is the only state accepting a word beginning with $e$. State $((m-1)',n-1)$ is the only empty state. 

For union we consider six possibilities for the distinguishability of two states:
\be
\item
States in $H$. If $n$ is odd, and $p',r'\in
Q'_n\setminus \{0', (n-1)'\}$  with $p < r$, where $p'$ and $r'$ are both final or both non-final, then $r'$ accepts $a^{n-1-r}$ whereas $p'$ rejects it. If $n$ is even, $r'$ accepts $a^{n-r}b$ whereas $p'$ rejects it.
Hence if $((m-1)',p)$ and $((m-1)',r)$ are in $H$, they are distinguishable.
\item
States in $V$.
The argument is symmetric to that in $H$.
\item
$H$ and $V$.
Consider $(p',n-1)\in  V$ and $((m-1)',q)\in  H$.
Note that $(1',n-1)$ accepts $a^2b$ whereas $((m-1)',1)$ rejects it.
For every $(p',q)$ there exists a permutation $t$ induced by a word $w$ in $\{a,b\}^*$ such that $(1',1)t=(p',q)$. Then $(p',q)t^{-1}=(1',1)$. Hence $(p',n-1)$ and $((m-1)',q)$ are mapped by $t^{-1}$ to $(1',n-1)$ and $((m-1)',1)$, which are distinguishable. 
\item
$H$ and $M$. If $(p',r)\in M$, then there exists a permutation $t$ induced by a word $w$ in $\{a,b\}^*$ such that $(p',r) t = (1',2)$. The same $t$ maps $((m-1)',q)\in H$ to some $((m-1)',s)$ with 
$1 \le s \le n-2$. Then $(1',2) d= ((m-1)',2)$. If $s \neq 1$, then $s d\in\{2,\dots,n-2)$, and we are done, since states $((m-1)',2)$ and $((m-1)',s)$ are in $H$ and hence are distinguishable. Otherwise, $((m-1)',s) d = ((m-1)',n-1)$ and again we are done.

\item
$V$ and $M$. The argument is symmetric to that for $H$ and $M$.
\item
States in $M$. Apply Theorem~1 of~\cite{BBMR14}.
\ee

For symmetric difference, we have as many states as for union. The arguments for distinguishability are exactly the same as for union.

For intersection, all states in $H\cup V\cup \{((m-1)',n-1)\}$ are empty, and all states in $M$ are distinguishable by~\cite[Theorem~1]{BBMR14}. Hence there are $(m-2)(n-2)+2=mn-2(m+n-3)$ states altogether.

For difference, all states in $H\cup \{((m-1)',n-1)\}$ are empty, so this leaves 
$(m-2)(n-1)+2=mn-(m+2n-4)$ states. 
The states in $V$ are all distinct by the argument used for union.
Also $(p',q) \in M$ is distinguishable from $(r',n-1)\in V$, by the argument used for union.
\end{proof}

Since $\delta_d=\delta_{c_1}\delta_{c_1}$ and $\delta_e=\delta_{c_1}\delta_{c_2}\cdots \delta_{c_{n-1}}$, where the $c_i$ are from Proposition~\ref{prop:W5}, the semigroup of $\cD_n(a,b,-,d,e)$ is in $\Vsf(n)\cap \Wsf(n)$.
 In fact, one can verify that the semigroup of $\cD_n(a,b,-,d,e)$ is $\Vsf(n)\cap \Wsf(n)$.

\subsection{Reversal}

Han and Salomaa~\cite{HaSa09} showed that to meet the bound for reversal one can use the binary DFA of Leiss~\cite{Lei81} and add a third input to get a suffix-free DFA.
Cmorik and Jir\'askov\'a~\cite{CmJi12} showed that a binary alphabet will not suffice.
We show a different ternary witness below, and prove that any witness must have its transition semigroup in $\Wsf$.

\begin{theorem}[Semigroup of Reversal Witness]
\label{thm:reversal_witness_subsemigroup}
For $n\ge 4$, the transition semigroup of a minimal DFA $\cD_n(Q_n,\Sig,\delta,0,F)$ of a suffix-free language $L_n$ that meets the bound $2^{n-2}+1$ for reversal is a subsemigroup of $\Wsf(n)$.
\end{theorem}
\begin{proof} 
It suffices to show that no pair of states is colliding.

We construct an NFA $\cN$ for $L_n^R$ by reversing the transitions of $\cD$, and interchanging the sets of final and initial states.
We then determinize $\cN$ using the subset construction to get a DFA $\cR$ for $L_n^R$. The states of $\cR$ are sets of states of $\cD$.

Consider a reachable subset $S \subseteq Q_n$. We can assume that $n-1 \not\in S$, since $n-1$ is not reachable from a start state in $\cN$.
Also, if $0 \in S$ then $S = \{0\}$, as otherwise in $\cD$ some transformation would map both $0$ and $q \in \{1,\ldots,n-1\}$ to  a final state so the language would not be suffix-free \cite[Lemma 6]{HaSa09}.
Hence, every subset $S \subseteq \{1,\ldots,n-2\}$ must be reachable.

Suppose that there are two distinct states $p,q \in \{1,\ldots,n-2\}$ such that the pair $\{p,q\}$ is colliding.
This means that $0 \delta_w = p$ and $q' \delta_w = q$ for some word $w$ and $q' \in \{1,\ldots,n-2\}$.
Since the subset $\{p,q\}$ is reachable in $T_{\cD^R}$, there is some $\delta^{-1}_u$ such that $F \delta^{-1}_u = \{p,q\}$.
So $p \delta^{-1}_u \in F$ and $q \delta^{-1}_u \in F$.
But then $w u \in L_n$, and also $v w u \in L_u$, where $v$ is such that $0 \delta_v = q'$, which contradicts that $L_n$ is suffix-free.
\end{proof}

\begin{theorem}[Reversal Complexity]
If $n\ge 4$, then $L_n(a,-,c,-, e)$ of Definition~\ref{def:D6} meets the bound $2^{n-2}+1$ for reversal.
\end{theorem}
\begin{proof}
As before, we construct NFA $\cN$ for $L^R$ -- state $n-1$ is now not reachable -- and determinize $\cN$ to get the DFA $\cR$. We shall show that we can reach $\{0\}$ and all subsets of $Q_n\setminus \{0,n-1\}$.

The initial state of $\cN$ consists of all the odd states in $Q_n\setminus \{0, n-1\}$.
By applying $e$ we reach the set $\{0\}$.

Consider a subset $S \subseteq Q_n\setminus \{0,n-1\}$.
By Proposition~\ref{prop:sf_wit}~(a) we have all transformations mapping $Q\setminus \{0,n-1\}$ into $Q\setminus \{0\}$;
in particular, if $n$ is odd, there is a transformation $\delta_w$ mapping $S$ into the initial subset $\{1,3,\dots,n-4,n-2\}$ and $Q_n \setminus S$ to $\{n-1\}$. Thus $\delta_w^{-1}$ maps $\{1,3,\dots,n-4,n-2\}$ onto $S$.
If $n$ is even, replace $\{1,3,\dots,n-4,n-2\}$ by $\{1,3,\dots,n-5,n-3\}$.

Now we prove that these $2^{n-2}+1$ states (sets of states of $\cD$) are pairwise distinguishable. Set $\{0\}$ is the only final state. Consider two sets  $R,S\subseteq Q_n\setminus \{0, n-1\}$ and suppose that $q\in R\oplus S$.  Then $a^{q-1}e$ is accepted from the set that contains $q$ but not from other set. Therefore there are no equivalent states.
\end{proof}

Although $L_n(a,-,c,-,e)$ meets the bound for the number of atoms, it does not meet the bounds on the quotient complexity of atoms.

\subsection{Complexity of Atoms in Suffix-Free Languages}
\label{sec:left}

Let $Q_n=\{0,\dots,n-1\}$ and let $L$ be a non-empty regular language with quotient set $K = \{K_0,\dotsc,K_{n-1}\}$. 
Let $\cD = (Q_n,\Sig,\delta,0,F)$ be the minimal DFA of $L$ in which the language of state $q$ is $K_q$.

Denote the complement of a language $L$ by $\ol{L} = \Sig^* \setminus L$.
Each subset $S$ of $Q_n$ defines an \emph{atomic intersection} $A_S = \bigcap_{i \in S} K_i \cap \bigcap_{i \in \ol{S}} \ol{K_i}$, where $\ol{S} = Q_n \setminus S$.
An \emph{atom} of $L$ is a non-empty atomic intersection. 
Since atoms are pairwise disjoint,  every atom $A$ has a unique atomic intersection associated with it, and this atomic intersection has a unique subset $S$ of $K$ associated with it. This set $S$ is called the \emph{basis} of $A$.
 
Let $A_S = \bigcap_{i \in S} K_i \cap \bigcap_{i \in \ol{S}} \ol{K_i}$ be an atom. For any $w\in\Sig^*$ we have 
$$w^{-1}A_S = \bigcap_{i \in S}w^{-1} K_i \cap \bigcap_{i \in \ol{S}} \ol{w^{-1} K_i}.$$
Since a quotient of a quotient of $L$ is also a quotient of $L$, $w^{-1}A_S$ has the form:
$$w^{-1} A_S = \bigcap_{i \in X} K_i \cap \bigcap_{i \in Y} \ol{K_i},$$
where $|X| \le |S|$ and $|Y| \le n-|S|$, $X,Y \subseteq Q_n$.

\begin{proposition}
\label{prop:atoms}
Suppose $L$ is a suffix-free language with $n\ge 4$ quotients. Then $L$ has at most $2^{n-2}+1$ atoms. 
Moreover, the complexity $\kappa(A_S)$ of atom $A_S$ satisfies

\begin{equation}
\label{eq:number_of_atoms}
	\kappa(A_S) 
	\begin{cases}
		 \le 2^{n-2}+1,		& \text{if $S=\emp$;} \\

		= n, 			& \text{if $S=\{0\}$;}\\
		 \le 1 + \sum_{x=1}^{|S|}\sum_{y=0}^{n-2-|S|}\binom{n-2}{x}\binom{n-2-x}{y},
		 			& \emp\neq S\subseteq \{1,\ldots,n-2\}.
		\end{cases}
\end{equation}
\end{proposition}
\begin{proof}
If ${n-1}\in S$, then $A_S$ is not an atom. Also,
$K_0\cap K_i=\emp$ for $i=1,\dots, n-1$, since $L$ is suffix-free. Hence $0\in S$ implies $S=\{0\}$.
It follows that there are at most  $2^{n-2}+1$ atoms.

For atom complexity, consider the following cases:

\be
\item
$S=\emp$. Here $\ol{S}=Q_n$, $A_\emp = \bigcap_{i \in Q_n} \ol{K_i}$
and for $w\in \Sig^+$, $w^{-1}A_S = \bigcap_{i \in Y} \ol{K_i},$ where $0\notin Y$, since $w^{-1}L =L$ only if $w=\eps$, and $n-1$ is always in $Y$ since $w^{-1}K_{n-1}=K_{n-1}$ because $K_{n-1}=\emp$.
 Thus we have the initial quotient $A_\emp$  with $0\in \ol{S}$ and at most $2^{n-2}$ choices for the other quotients of $A_S$. Hence the complexity of $A_{\emp}$ is at most $2^{n-2}+1$.
\item
$S=\{0\}$.
Since the language is suffix-free, we have that $K_0 \cap K_q = \emptyset$ for all $q \in Q \setminus \{0\}$.
Thus $A_0 = K_0 \cap \ol(K_1) \cap \ldots \cap \ol(K_{n-1}) = K_0$, and $K_0 = L_n$ has $n$ quotients.
\item
Since $n-1$ always appears in $\ol{S}$ and the $S = \{0\}$ case is done, there remain the cases where $\emp \ne S \subseteq Q_n \setminus \{0, n-1\}$.
Suppose $(X,Y \cup \{n-1\})$ represents a quotient of $A_S$ by a non-empty word $w$.
If $0 \delta_w = n-1$, then $X$ must have at least one and at most $|S|$ elements from $Q_n \setminus \{0, n-1\}$. Since $0$ appears only initially, $Y$ cannot contain $0$. So in $Y$ there must be from $0$ to $n-2-|X|$ states from $Q_n \setminus (\{0,n-1\} \cup X)$. Hence we have the formula from Equation~\ref{eq:number_of_atoms}.
Suppose that $0 \delta_w \neq n-1$. The $0$-path in $\delta_w$ is aperiodic and ends in $0$, so there exists a state $p \in Q_M$ such that $p \delta_w = n-1$. If $p \in S$, then $n-1 \in X$, and so $(X,Y)$ represents the empty quotient.
If $p \in Q \setminus S$, then again $Y \setminus \{n-2,n-1\}$ contains at most $n-2-|S|$ states.
\ee
\end{proof}

Following Iv\'an~\cite{Iva15} we define a DFA for each atom:

\begin{definition}
Suppose $\cD=(Q,\Sig,\delta, q_0, F)$ is a DFA and let $S \subseteq Q$.
Define the DFA $\cD_S = (Q_S,\Sig,\Delta,(S,\ol{S}),F_S)$, where
\bi
\item
$Q_S = \{(X,Y) \mid X,Y \subseteq Q, X \cap Y = \emp\} \cup \{\bot\}$.
\item
For all $a \in \Sig$, $(X,Y)a = (Xa,Ya)$ if $Xa \cap Ya = \emp$, and $(X,Y)a = \bot$ otherwise; and $\bot a = \bot$.
\item
$F_S = \{(X,Y) \mid X\subseteq F, Y \subseteq \ol{F}\}$. 
\ei
\end{definition}
DFA $\cD_S$ recognizes the atomic intersection $A_S$ of $L$. If $\cD_S$ recognizes a non-empty language, 
then $A_S$ is an atom.

\begin{theorem}
For $n\ge 4$, the language $L_n(\cD_n(a,b,c,d,e))$ of Definition~\ref{def:D6} meets the bounds of Proposition~\ref{prop:atoms} for the atoms. 
\end{theorem}
\begin{proof}
There are three cases to consider:
\be
\item $S=\emp$. The initial state of $\cD_S$ is $(\emp,Q_n)$. Since $(\emp,Q_n)ed=
(\emp,\{n-1\})$, which is a final state, $A_\emp$ is an atom. For any non-empty word $w$, we have a quotient of $A_\emp$ represented by $(\emp,Y)$, where $0\notin Y$, $n-1\in  Y$, and $Y\setminus \{n-1\}$ is a subset of $Q_n\setminus \{0,n-1\}$.
By Proposition~\ref{prop:sf_wit} the transition semigroup  of $\cD$ contains all transformations 
that send $0$ and $n-1$ to $n-1$ and map 
$Q_n\setminus \{0,n-1\}$ into $Q_n\setminus \{0\}$; hence
all $2^{n-2}$ quotients of $A_\emp$ by non-empty words are reachable.
Together with $A_\emp$ we have $2^{n-2}+1$ quotients of $A_\emp$.

Since $(\emp, Q_n) e = (\emp, \{1, n-1\})$, which is non-final,  and 
$(\emp,Y) e =(\emp, \{n-1\})$, which is final, $(\emp,Y)$ is distinguished from the initial state.
We claim also that all pairs $(\emp,Y)$ and $(\emp,Y')$, $Y\neq Y'$, are distinguishable. 
 Without loss of generality, suppose that $q\in Y\setminus Y'$. If we map $q$ to $1$ and $(Y\cup Y')\setminus \{1\}$ to $n-1$ by a word $w$, 
then $(\emp,Y)w =(\emp,\{1,n-1\})$, which is non-final, and 
$(\emp,Y')w=(\emp,\{n-1\})$, which is final.
Hence the bound $2^{n-2}+1$ is met.

\item  $S=\{0\}$. The initial state of $\cD_S$ is $(\{0\},Q_n\setminus \{0\})$.
We have $(\{0\},Q_n\setminus \{0\})x=\bot$ for all $x\in\Sig\setminus\{e\}$, 
and $(\{0\},Q_n\setminus \{0\})e=(\{1\},\{n-1\})$.
From $(\{1\},\{n-1\})$ we can reach any state $(\{q\},\{n-1\})$ with $q\in Q_n \setminus \{0,n-1\}$. 
Thus we can reach $n$ states altogether: the initial state, $n-2$ states $(\{q\},\{n-1\})$ with $q\in Q_n \setminus \{0,n-1\}$, and $\bot$.

Since $(\{0\},Q_n\setminus \{0\})e=(\{1\},\{n-1\})$, which is final,  the initial state is non-empty and so differs from $\bot$.
For $q \in Q_n \setminus \{0,n-1\}$,  consider the word $w$ that maps $q$ to $1$ and $Q_n\setminus \{0, q\}$ to $\{n-1\}$.
Then $(\{ q \},\{n-1\}) w = (\{1\},\{n-1\})$, which is final, while 
$(\{0\},Q_n\setminus \{0\})w=(\{n-1\},Y) $, which is non-final since $Y$ contains $n-1$.
Hence $(\{ q \},\{n-1\})$ is distinguishable from the initial state and $\bot$.
Finally,  for $p,q \in Q_n \setminus \{0,n-1\}$,  $(\{ p \},\{n-1\})$ is distinguishable from  $(\{ q \},\{n-1\})$ by the mapping that sends $p$ to $1$ and $q$ to $n-1$.
Hence the bound $n$ is met.

\item $\emp\neq S \subseteq Q_n\setminus\{0,n-1\}$.
The initial state here is $(S,\ol{S})$, where $ \emp\neq S\subseteq \{1,\ldots,n-2\}$.
Since we can map $S$ to $\{1\}$ and $\ol{S}$ to $\{n-1\}$, each such intersection is an atom.
If $M=\{1,\dots,n-2\}$, the atom has the form
$(S,\ol{S}) = (S,(M\setminus S)\cup\{0, n-1\})$.
By applying $e$, we get $\bot$.
To get a state of $A_S$ different from $\bot$, we can map $S$ to any non-empty subset $X$ of $M$ of cardinality 
$|X|$, where $1 \le |X| \le |S|$;  we can map $M\setminus S$ to any subset
$Y$ of $(M\cup \{n-1\})\setminus X$ of cardinality $|Y|$, where $0 \le |Y| \le |M|-|X|$; and we can map $0$ and $n-1$ to $n-1$.
Each such intersection $(X,Y\cup \{n-1\})$ represents a non-empty quotient of $A_S$ by a non-empty word because we can map $X$ to $\{1\}$ and $Y$ to $\{n-1\}$.
Hence we have the formula in Equation~\ref{eq:number_of_atoms} for the number of reachable states of $\cD_S$.

Now consider two states $(X,Y\cup \{n-1\})$ and $(X',Y'\cup \{n-1\})$, where one of these two states could be the initial one.
If $X\neq X'$, assume that $q\in X'\setminus X$.
Let $w$ map $X$ to $\{1\}$ and $Q_n\setminus X$ to $n-1$.
Then  $(X,Y\cup \{n-1\}) w =(\{1\},\{n-1\})$, while 
the first component of $(X',Y'\cup \{n-1\}) w $ contains $n-1$, meaning that that state is $\bot$; so these states are distinguishable. 
A~symmetric argument applies if $q\in X\setminus X'$.
 
This leaves the case where $X=X'$. Then $X$ is disjoint from both $Y$ and $Y'$.
If $q\in Y' \setminus Y$, let $w$ map $X$ to $\{1\}$,  
$Y\cup \{0\}$ to $\{n-1\}$, and $Y'\setminus Y$ to $\{1\}$. 

Then $(X,Y\cup \{n-1\}) w =(\{1\},\{n-1\})$, while $(X,Y'\cup \{n-1\}) w =(\{1\},\{1,n-1\})=\bot$, and we have distinguishablity.
A~symmetric argument applies if $q\in Y\setminus Y'$.
\ee
\end{proof}

\begin{theorem}
The transition semigroup $T_n$ of a minimal DFA $\cD_n(Q_n,\Sig,\delta,0,F)$ with $n \ge 4$ of a suffix-free language $L_n$ that meets the bounds for atom complexities from Proposition~\ref{prop:atoms} is a subsemigroup of $\Wsf(n)$ and is not a subsemigroup of $\Vsf(n)$.
\end{theorem}
\begin{proof}
It is sufficient to show that every pair of states $p,q \in \{1,\ldots,n-2\}$ is focused by some transformation from $T_n$.

Let $p,q \in \{1,\ldots,n-2\}$ be two distinct states.
Consider the atom $A_{\{p,q\}}$.
From the proof of Proposition~\ref{prop:atoms}, to meet the bound for the quotient complexity of $A_{\{p,q\}}$ every pair $(X,Y \cup \{n-1\})$ must represent a quotient of the atom by a non-empty word $w$, where $X,Y \subseteq \{1,\ldots,n-2\}$ are disjoint, $1 \le |X| \le 2$, and $|Y| \le n-2-|X|$.
In particular, when $|X|=1$ we know that the transformation of $w$ maps both $p$ and $q$ to the state from $X$, and so $p$ and $q$ are focused.
\end{proof}

\begin{remark}
The complexity of atoms in left ideals~\cite{BrDa15} is 
\begin{equation}
	\kappa(A_S) 
	\begin{cases}
		= n, 			& \text{if $S=Q_n$;}\\
		\le 2^{n-1},		& \text{if $S=\emp$;}\\
		 \le 1 + \sum_{x=1}^{|S|}\sum_{y=1}^{n-|S|}\binom{n-1}{x}\binom{n-1-x}{y-1},
		 			& \text{otherwise.}
		\end{cases}
\end{equation}
The formula for $S\not\in \{\emp, Q_n\}$ evaluated for  $n-1$  and $S\subseteq \{1,\dots,n-2\}$ becomes
$$1 + \sum_{x=1}^{|S|}\sum_{y=1}^{n-2-|S|}\binom{n-2}{x}\binom{n-2-x}{y-1},$$
which is precisely the formula for suffix-free languages.\qedb
\end{remark}

Tables~\ref{tab:atom_comp1} and~\ref{tab:atom_comp2} show the quotient complexities of atoms for small $n$.

\renewcommand{\arraystretch}{1.3}
\begin{table}[ht]
\caption{Suffix-free atom complexity. The entries from left to right are suffix-free language, left ideal, regular language. The $\ast$ stands for ``not applicable''.}\label{tab:atom_comp1}
\footnotesize
\begin{center}
$
\begin{array}{|c|c|c|c|c|c|c|}
\hline
\ n\ & 1\ &\ 2\ &\ 3\ &\ 4\ &\ 5\ &\ \cdots\\
\hline
\hline
|S|=0 & \ast/{\bf 1}/{\bf 1} &\ast/{\bf 2}/{\bf 3} &\ast/4/7 & 5 / 8/15 & 9/16/31 & \cdots\\
\hline\hline
|S|=1 & \ast /{\bf 1}/{\bf 1} &\ast /{\bf 2}/{\bf 3} &\ast /{\bf 5}/{\bf 10} &{\bf 5}/13/29 & 13/33/76 & \cdots\\
\hline
|S|=2 &  & \ast /{\bf 2}/{\bf 3} &\ast /4/{\bf 10} &4/{\bf 16}/{\bf 43} &{\bf 16} /{\bf 53}/{\bf 141} &\cdots\\
\hline
|S|=3 &  &  & \ast /3/7 & \ast/8/29 & 8/ 43/{\bf 141} &\cdots\\
\hline
|S|=4 &  &  &  & \ast/4/15 & \ast/16/76 & \cdots\\
\hline
|S|=5 &  &  &  &  & \ast/5/31 & \cdots\\
\hline
\textit{max} & \ast/1/1 & \ast/2/3 & \ast/5/10 & 5/16/43 & 16/53/141 & \cdots\\
\hline
\textit{ratio} & - & \ast/2.00/3.00 & \ast/2.50/3.33 & \ast/3.20/4.30 & 3.20/3.31/3.28 & \cdots\\
\hline
\end{array}
$
\end{center}
\end{table}
\renewcommand{\arraystretch}{1.3}
\begin{table}[ht]
\caption{Suffix-free atom complexity continued.}
\label{tab:atom_comp2}
\footnotesize
\begin{center}
$
\begin{array}{|c|c|c|c|c|}
\hline
\hline
\ n\ & 6\ &\ 7\ &\ 8\ &\ 9\ \\
\hline
\hline
|S|=0 & 17/32/63 &33/ 64/127 & 65/ 128/255 & 129/ 256/511 \\
\hline\hline
|S|=1 & 33/81/187 & 81/193/442 & 193/449/1,017 & 449/1,025/2,296 \\
\hline
|S|=2 & {\bf 53}/ 156/406 & 156/ 427/1,086 & 427/ 1,114/2,773 & 1,114/ 2,809/6,859 \\
\hline
|S|=3 & 43/ {\bf 166}/{\bf 501} &{\bf 166}/ {\bf 542}/{\bf 1,548} &{\bf 542}/ 1,611/4,425 & 1,611/ 4,517/12,043 \\
\hline
|S|=4 & 16/ 106/406 & 106/ 462/{\bf 1,548} &462/ {\bf 1,646}/{\bf 5,083} &{\bf 1,646}/ {\bf 5,245}/{\bf 15,361} \\
\hline
|S|=5 & \ast/32/187 & 32/ 249/1,086 & 249/ 1,205/4,425 & 1,205/ 4,643/{\bf 15,361} \\ 
\hline
|S|=6 & \ast/6/63 & \ast /64/442 & 64/ 568/2,773 & 568/ 3,019/12,043 \\
\hline
|S|=7 &  & \ast/7/127 & \ast/128/1,017 & 128/1,271/6,859 \\
\hline
|S|=8 &  &  & \ast/8/255 & \ast /256/2,296 \\
\hline
|S|=9 &  &  &  & \ast/9/511 \\
\hline
\textit{max} & 53/166/501 & 166/542/1,548 & 542/1,646/5,083 & 1,646/5,245/15,361 \\
\hline
\textit{ratio} &  3.31/3.13/3.55 & 3.13/3.27/3.09 & 3.27/3.04/3.28 & 3.04/3.19/3.02 \\
\hline
\end{array}
$
\end{center}
\end{table}

\section{Conclusions}
It may appear that the semigroup $\Vsf(n)$ should not be of great importance, since it exceeds $\Wsf(n)$ only for $n=4$ and $n=5$, and therefore should not matter when $n$ is large. However, our results show that this is not the case. 
We conclude with our result about the non-existence of single universal suffix-free witness.

\begin{theorem}
For $n\ge 4$, there does not exist a most complex suffix-free language.
\end{theorem}
\begin{proof}
If there exists a most complex DFA $\cD(Q_n,\Sig,\delta,0,F)$, then in particular it would have to meet the bound for reversal and star.
From Theorem~\ref{thm:reversal_witness_subsemigroup} the transition semigroup of $\cD$ must be a subsemigroup of $\Wsf(n)$, but from  Theorem~\ref{thm:star_witness_subsemigroup} we know that it cannot be a subsemigroup of $\Wsf(n)$.
Consequently no most complex suffix-free language exists for $n \ge 4$.
\end{proof}

The first four studies of most complex languages were done for the classes of regular languages~\cite{Brz13}, right ideals~\cite{BrDa14,BrDa15,BrSiJALC}, left ideals~\cite{BrDa15,BDL17,BrSiJALC}, and two-sided ideals~\cite{BrDa15,BDL17,BrSiJALC}. In those cases there exists a single witness over a minimal alphabet which, together with its permutational dialects, covers all the complexity measures. In the case of suffix-free languages such a witness does not exist. Our study is an example of a general problem: Given a class of regular languages, find the smallest set of witnesses over minimal alphabets that together cover all the measures. 
The witness of Definition~\ref{def:D5} is conjectured to be over a minimal alphabet, unless the  bound for product can be met by binary DFAs for every $n,m > c$, for some $c$; this is an open problem.
The witness of Definition~\ref{def:D6} is over a minimal alphabet, since five letters are required to meet the bound for syntactic complexity.
\medskip

\noin
{\bf Acknowledgments}

\noin This work was supported by the Natural Sciences and Engineering Research Council of Canada (NSERC)
grant No.~OGP000087, and by the National Science Centre, Poland under grant number 2014/15/B/ST6/00615.

\noin We are grateful to Sylvie Davies for careful proofreading.

\bibliographystyle{plain}
\providecommand{\noopsort}[1]{}

\end{document}